\documentclass[a4paper,UKenglish,thm-restate,cleveref]{lipics-v2021}

\pdfoutput=1 %
\hideLIPIcs  %

\usepackage{booktabs}

\title{Query maintenance under batch changes with small-depth circuits}

\author{Samir Datta}{Chennai Mathematical Institute  \& UMI ReLaX, Chennai, India}{sdatta@cmi.ac.in}{https://orcid.org/0000-0003-2196-2308}{Partially funded by a grant from Infosys foundation.}
\author{Asif Khan}{Chennai Mathematical Institute, Chennai, India}{asifkhan@cmi.ac.in}{}{Partially funded by a grant from Infosys foundation.}
\author{Anish Mukherjee}{University of Warwick, Coventry, United Kingdom}{anish.mukherjee@warwick.ac.uk}{https://orcid.org/0000-0002-5857-9778}{Research was supported in part by the Centre for Discrete Mathematics and its Applications (DIMAP) and by EPSRC award EP/V01305X/1.}
\author{Felix Tschirbs}{Ruhr University Bochum, Bochum, Germany}{felix.tschirbs@rub.de}{}{Supported by the Deutsche Forschungsgemeinschaft (DFG, German
Research Foundation), grant 532727578.}
\author{Nils Vortmeier}{Ruhr University Bochum, Bochum, Germany}{nils.vortmeier@rub.de}{https://orcid.org/0009-0000-2821-7365}{Supported by the Deutsche Forschungsgemeinschaft (DFG, German
Research Foundation), grant 532727578.}
\author{Thomas Zeume}{Ruhr University Bochum, Bochum, Germany}{thomas.zeume@rub.de}{https://orcid.org/0000-0002-5186-7507}{Supported by the Deutsche Forschungsgemeinschaft (DFG, German
Research Foundation), grant 532727578.}

\authorrunning{S. Datta, A. Khan, A. Mukherjee, F. Tschirbs, N. Vortmeier, and T. Zeume} 

\Copyright{Samir Datta, Asif Khan, Anish Mukherjee, Felix Tschirbs, Nils Vortmeier, and Thomas Zeume}



\category{} %

\nolinenumbers %

\EventEditors{Rastislav Kr\'{a}lovi\v{c} and Anton\'{i}n Ku\v{c}era}
\EventNoEds{2}
\EventLongTitle{49th International Symposium on Mathematical Foundations of Computer Science (MFCS 2024)}
\EventShortTitle{MFCS 2024}
\EventAcronym{MFCS}
\EventYear{2024}
\EventDate{August 26--30, 2024}
\EventLocation{Bratislava, Slovakia}
\EventLogo{}
\SeriesVolume{306}
\ArticleNo{38}

\usepackage{graphicx}
\usepackage{xspace}
\usepackage{enumerate}
\usepackage{tikz}
\usetikzlibrary{arrows,decorations.pathmorphing,backgrounds,calc,positioning,fit,petri,matrix,backgrounds, decorations.pathreplacing, shapes.geometric}
\usepackage[ruled]{algorithm}
\usepackage{algpseudocode}

\newif\ifcomments
\newif\ifchanges
\commentsfalse\changesfalse
\commentstrue\changestrue %

\newcommand{\superpolylog}{\ensuremath{(\log n)^{\bigO(\log \log n)}}}
\newcommand{\plogtable}{\ensuremath{(\log n)^{\bigO(1)}}}

\newcommand{\tpl}{\bar}

\newcommand{\mtext}[1]{\textsc{#1}}
\newcommand{\dist}{\mtext{dist}\xspace}

\newcommand{\ins}{\mtext{ins}\xspace}
\newcommand{\del}{\mtext{del}\xspace}

\newcommand{\schema}{\ensuremath{\sigma}\xspace}

\newcommand{\query}{\ensuremath{Q}}
\newcommand{\inp}{\ensuremath{\calI}\xspace}
\newcommand{\aux}{\ensuremath{\calA}\xspace}%

\newcommand{\struc}{\calS}
\newcommand{\univ}{D}

\newcommand{\restrict}[2]{#1[#2]}
\newcommand{\N}{\ensuremath{\mathbb{N}}}

\newcommand{\bigO}{\ensuremath{\mathcal{O}}}

\newcommand{\df}{\ensuremath{\mathrel{\smash{\stackrel{\scriptscriptstyle{
    \text{def}}}{=}}}} \;}

\DeclareMathOperator{\polylog}{polylog}

\newcommand  {\myclass} [1]  {\ensuremath{\textsf{\upshape #1}}}

\newcommand{\StaClass}[1]{\myclass{#1}\xspace}

\newcommand{\DynClass}[1]{\myclass{Dyn#1}\xspace}

\newcommand     {\LOGSPACE}     {\StaClass{LOGSPACE}}
\newcommand     {\NC}   {\StaClass{NC}}

\newcommand     {\AC}   {\StaClass{AC}}
\newcommand     {\Ind}   {\StaClass{IND}}

\newcommand     {\SAC}   {\myclass{SAC}}

\newcommand{\ACz}{\mbox{\myclass{AC}$^0$}\xspace}
\newcommand{\TCz}{\mbox{\myclass{TC}$^0$}\xspace}
\newcommand{\TCo}{\mbox{\myclass{TC}$^1$}\xspace}

\newcommand{\FO}{\StaClass{FO}}

\newcommand{\FOar}{\StaClass{FO$(\leq,+,\times)$}}%
\newcommand{\FOLL}{\StaClass{FOLL}}
\newcommand{\MSO}{\StaClass{MSO}}

\newcommand{\CQ}[1][]{\StaClass{CQ}}
\newcommand{\UCQ}[1][]{\StaClass{UCQ}}
\newcommand{\CQneg}[1][]{\StaClass{CQ\ensuremath{^{\mneg}}}}
\newcommand{\UCQneg}[1][]{\StaClass{UCQ\ensuremath{^{\mneg}}}}

\newcommand{\mneg}{\neg} %

\newcommand{\DynFO}{\DynClass{FO}}
\newcommand{\DynAC}{\DynClass{AC}}
\newcommand{\DynFOLL}{\DynClass{FOLL}}

\newcommand{\DynFOar}{\DynClass{FO$(\leq,+,\times)$}}
\theoremstyle{definition}
\newtheorem*{question*}{Question}
\newtheorem*{openquestion*}{Open question}

\providecommand {\calA}      {{\mathcal A}\xspace}

\providecommand {\calI}      {{\mathcal I}\xspace}

\providecommand {\calR}      {{\mathcal R}\xspace}
\providecommand {\calS}      {{\mathcal S}\xspace}

\providecommand {\calX}      {{\mathcal X}\xspace}

\newcommand{\prog}{\ensuremath{\Pi}\xspace}

\newcommand{\auxSchema}{\ensuremath{\schema_{\text{aux}}}\xspace}

\ifcomments
\newcommand{\commentbox}[1]{\noindent\framebox{\parbox{0.98\linewidth}{#1}}}

\newcommand{\acomment}[2]{\ \\ \fbox{\parbox{0.98\linewidth}{{\sc #1}: #2}}}
\newcommand{\mcomment}[2]{{\color{blue}(#1)}\footnote{#1: #2}} %
\else
\newcommand{\commentbox}[1]{}
\newcommand{\mcomment}[2]{}
\newcommand{\acomment}[2]{}
\fi

\ifchanges

\setul{}{0.2mm}
\setstcolor{red}

\else

\fi

\definecolor{iltisBeige1}{HTML}{fef6ee}
\definecolor{iltisBeige2}{HTML}{e3d5c8}
\definecolor{iltisBeige3}{HTML}{cfbeb0}
\definecolor{iltisBeige4}{HTML}{bfac9b}

\definecolor{iltisGrey1}{HTML}{edebe8}
\definecolor{iltisGrey2}{HTML}{d9d6d2}
\definecolor{iltisGrey3}{HTML}{c2bfbc}
\definecolor{iltisGrey4}{HTML}{a6a4a1}

\definecolor{iltisLightGreen1}{HTML}{f4ffe9}
\definecolor{iltisLightGreen2}{HTML}{d4f7b2}
\definecolor{iltisLightGreen3}{HTML}{bde697}
\definecolor{iltisLightGreen4}{HTML}{a3d177}

\definecolor{iltisGreen1}{HTML}{cfffe1}
\definecolor{iltisGreen2}{HTML}{a2e8bd}
\definecolor{iltisGreen3}{HTML}{86d1a2}
\definecolor{iltisGreen4}{HTML}{6fbf8d}

\definecolor{iltisYellow1}{HTML}{fef2d0}
\definecolor{iltisYellow2}{HTML}{ffe3a2}
\definecolor{iltisYellow3}{HTML}{ffd77d}
\definecolor{iltisYellow4}{HTML}{f2c55e}

\definecolor{iltisRed1}{HTML}{ffe9e6}
\definecolor{iltisRed2}{HTML}{eda498}
\definecolor{iltisRed3}{HTML}{e08475}
\definecolor{iltisRed4}{HTML}{c26e60}

\definecolor{iltisOrange1}{HTML}{ffdfb3}
\definecolor{iltisOrange2}{HTML}{ffc97d}
\definecolor{iltisOrange3}{HTML}{ebb467}
\definecolor{iltisOrange4}{HTML}{e0a34c}

\definecolor{iltisCyan1}{HTML}{e0fffe}
\definecolor{iltisCyan2}{HTML}{b4e0df}
\definecolor{iltisCyan3}{HTML}{95c7c5}
\definecolor{iltisCyan4}{HTML}{81b3b0}

\definecolor{iltisBlue1}{HTML}{cce8ff}
\definecolor{iltisBlue2}{HTML}{8db8d9}
\definecolor{iltisBlue3}{HTML}{6f9abd}
\definecolor{iltisBlue4}{HTML}{5c86a8}
\definecolor{iltisBlue5}{HTML}{2e5b80}

\definecolor{iltisViolet1}{HTML}{ede6ff}
\definecolor{iltisViolet2}{HTML}{b1a5cc}
\definecolor{iltisViolet3}{HTML}{9b8eba}
\definecolor{iltisViolet4}{HTML}{8578a6}

\colorlet{midgray}{black!50}

\tikzstyle{dEdge}=[
  -latex', %
  thick, 
  shorten >=2pt, 
  shorten <=2pt,
  draw=black!80,
]

\tikzstyle{dhEdge}=[
  -latex', %
  thick, 
  shorten >=3pt, 
  shorten <=3pt,
  draw=black!80,
]
\tikzstyle{uEdge}=[
  thick, 
  shorten >=3pt, 
  shorten <=3pt,
  draw=black!80,
]
\tikzstyle{uhEdge}=[
  thick, 
  shorten >=3pt, 
  shorten <=3pt,
  draw=black!80,
]

\tikzstyle{cEdge}=[
  ultra thick, 
  shorten >=3pt, 
  shorten <=3pt,
  draw=black!80,
]

\tikzstyle{dotsEdge}=[
  very thick, 
  loosely dotted, 
  shorten >=3pt, 
  shorten <=3pt
]

\tikzstyle{snakeEdge}=[
  ->, 
  decorate, 
    shorten >=2pt, 
  shorten <=2pt,
  decoration={snake,amplitude=.4mm,segment length=2.5mm,post length=0.0mm},
]

\tikzstyle{snakeEdgea}=[
  ->, 
  decorate, 
  decoration={snake,amplitude=.4mm,segment length=3mm,post length=0.5mm}
]
\tikzstyle{blackNode}=[
	shape=circle, draw, fill,inner sep=0pt,minimum size=4pt
]

\begin{document}

\maketitle

\begin{abstract}
Which dynamic queries can be maintained efficiently? For constant-size changes, it is known that constant-depth circuits or, equivalently, first-order updates suffice for maintaining many important queries, among them reachability, tree isomorphism, and the word problem for context-free languages. In other words, these queries are in the dynamic complexity class \DynFO. We show that most of the existing results for constant-size changes can be recovered for batch changes of polylogarithmic size if one allows circuits of depth $\bigO(\log \log n)$ or, equivalently, first-order updates that are iterated $\bigO(\log \log n)$ times.
\end{abstract}

\section{Introduction}\label{section:introduction}
Dynamic descriptive complexity \cite{PatnaikI97, DongS93} is a framework for studying the amount of resources that are necessary to \emph{maintain} the result of a query when the input changes slightly, possibly using additional auxiliary data (which needs to be maintained as well). Its main class \DynFO contains all queries for which the update of the query result (and possibly of further useful auxiliary data) can be expressed in first-order logic \FO. Equivalently\footnote{assuming that first-order formulas have access to numeric predicates $\leq, +, \times$}, the updates can be computed using (DLOGTIME) uniform circuits with constant-depth and polynomial size that consist of $\neg\hspace{0.2mm}$- as well as $\wedge$- and $\vee$-gates with unbounded fan-in, that is, within uniform~\ACz.

It is known that many important queries can be maintained in \DynFO if only one bit of the input changes in every step. This includes reachability for acyclic graphs \cite{DongS95, PatnaikI97}, undirected graphs \cite{PatnaikI97, DongS98, GradelS12}, and general directed graphs \cite{DattaKMSZ18}, tree isomorphism \cite{Etessami98} and every problem definable in monadic second-order logic \MSO for graphs of bounded treewidth~\cite{DattaMSVZ19}, all under insertions and deletions of single edges. Also, membership in context-free languages can be maintained under changes of single positions of the input word \cite{GeladeMS12}.

Some of these results have been extended to changes beyond single-bit changes: reachability in undirected graphs is in \DynFO if simultaneously $\polylog(n) = \plogtable$ edges can be inserted or deleted \cite{DattaKMTVZ20}, where $n$ is the size of the graph; regular languages are in \DynFO under changes of $\polylog(n)$ positions at once \cite{TschirbsVZ23}. Reachability in directed graphs can be maintained under insertions and deletions of $\bigO(\frac{\log n}{\log \log n})$ many edges \cite{Datta0VZ18}. 

Thus, only for few problems it is known that changes of polylogarithmic size (or: even non-constant size) can be handled in $\DynFO$, or, equivalently, by $\AC^0$-updates.
Trivially, if a problem can be maintained in \DynFO under single-bit changes it can also be maintained under $\polylog(n)$ changes using \AC-circuits of $\polylog(n)$-depth. This is achieved by processing the changed bits ``sequentially'' by ``stacking'' $\polylog(n)$ copies of the constant-depth circuit for processing single-bit changes.

The starting point for the present paper is the question which problems can be maintained by \AC-circuits of less than $\polylog(n)$ depth under $\polylog(n)$-sized changes, in particular which of the problems  known to be in $\DynFO$ under single-bit changes. The answer is short: for almost all of them circuits of depth $\bigO(\log\log n)$ suffice.

A first observation is that directed reachability under polylogarithmic changes can be maintained by $\AC$-circuits of depth $\bigO(\log \log n)$. This can be derived by analyzing the proof from \cite{Datta0VZ18} (see Section \ref{section:small-structures}). For this reason, we introduce the dynamic complexity class\footnote{The class could equally well be called (uniform) $\DynAC[\log\log n]$. We opted for the name $\DynFOLL$ as it extends $\DynFO$ and its static variant was introduced as \FOLL \cite{BarringtonKLM01}.} \DynFOLL of problems that can be maintained using circuits with polynomial size and depth $\bigO(\log \log n)$ or, equivalently, by first-order formulas that are iterated $\bigO(\log \log n)$ times. We investigate its power when changes affect $\polylog(n)$ input bits and prove that almost all problems known to be maintainable in \DynFO for constant-size changes fall into this class for changes of $\polylog(n)$-size, see Table \ref{table:overview}. One important problem left open is whether all \MSO-definable queries for bounded treewidth graphs can be maintained in \DynFOLL under $\polylog(n)$ changes. We present an intermediate result and show that tree decompositions can be maintained within \DynFOLL (see Section~\ref{section:bodlaender}).

This power of depth-$\bigO(\log\log n)$ update circuits came as a surprise to us. Statically, circuits of this depth and polynomial size still cannot compute the parity of $n$ bits due to H\aa{}stad's famous lower bound for parity: depth-$(d+1)$ \AC-circuits with alternating $\wedge$- and $\vee$-layers require $2^{\Omega(n^{1/d})}$ gates for computing parity (see, e.g., \cite[Theorem 12.3]{Jukna2012}). Dynamically, while such update circuits are powerful for changes of non-constant size, they seem to provide not much more power for single-bit changes. As an example, the parity-exists query from \cite{VortmeierZ21} is conjectured to not be in $\DynFO$, and it also cannot easily be seen to be in \DynFOLL.

The obtained bounds are almost optimal. For all mentioned problems, \DynFO can handle changes of size at most $\polylog(n)$ and $\DynFOLL$ can handle changes of size at most $\superpolylog$. This is an immediate consequence of H\aa{}stad's lower bound for parity  and standard reductions from parity to these problems. For the queries that are known to be maintainable under $\polylog(n)$ changes in \DynFO, we show that they can be maintained under $\superpolylog$ changes in \DynFOLL.

\begin{table}[t]
\centering
\begin{tabular}{@{}l|l|l@{}}
 Dynamic query &  \multicolumn{1}{l|}{\DynFO} & \multicolumn{1}{l}{\DynFOLL} \\ \hline
reachability  & & \\
\hspace{4mm} general graphs & $\bigO(\frac{\log n}{\log \log n})\quad$ \hfill \cite{Datta0VZ18} & $\plogtable$ \hfill(Theorem \ref{thm:reachMain})\\
\hspace{4mm} undirected graphs & $\plogtable$ \hfill\cite{DattaKMTVZ20} & $\superpolylog \quad$ (Theorem \ref{thm:undirected}) \\
\hspace{4mm} acyclic graphs & $\bigO(\frac{\log n}{\log \log n})$ \hfill\cite{Datta0VZ18} & $\plogtable$ \hfill(Theorem \ref{thm:reachMain})\\
\hline
distances & &  \\
\hspace{4mm} general graphs & open & open \\
\hspace{4mm} undirected graphs & $\bigO(1)$ \hfill\cite{GradelS12} & $\plogtable$ \hfill(Theorem \ref{thm:distances})\\
\hspace{4mm} acyclic graphs & $\bigO(1)$  & $\plogtable$ \hfill(Theorem \ref{thm:distances})\\
\hline
bounded tree width & &  \\
\hspace{4mm} tree decomposition & open & $\plogtable$ \hfill(Theorem \ref{theorem:bodlaender})\\
\hspace{4mm} MSO properties & $\bigO(1)$ \hfill \cite{DattaMSVZ19}& $\bigO(1)$ \\
\hline
other graph problems & &  \\
\hspace{4mm} tree isomorphism & $\bigO(1)$ \hfill\cite{Etessami98} &  $\plogtable$ \hfill(Theorem \ref{thm:tiso:dynfoll})\\
\hspace{4mm} minimum spanning forest & $\plogtable\quad$ \hfill(Theorem \ref{thm:small-structure-further}) & $\superpolylog \quad$ (Theorem \ref{thm:small-structure-further}) \\
\hspace{4mm} maximal matching & $\bigO(1)$ \hfill \cite{PatnaikI97} & $\plogtable \quad$ \hfill (Theorem \ref{thm:small-structure-further}) \\
\hspace{4mm} $(\delta+1)$-colouring & $\plogtable$ \hfill(Theorem \ref{thm:small-structure-further}) & $n^2 \quad$ \hfill (static, \cite{GoldbergP87}) \\
\hline
word problem & &  \\
\hspace{4mm}regular languages & $\plogtable$ \hfill\cite{TschirbsVZ23} &  $\superpolylog$ \hfill (Theorem \ref{thm:undirected})\\
\hspace{4mm}context-free languages & $\bigO(1)$ \hfill\cite{GeladeMS12}& $\plogtable$ \hfill(Theorem \ref{thm:cfl_dynfoll})\\
\end{tabular}\vspace{2mm}

\caption{Overview of results for \DynFO and \DynFOLL. Entries indicate the size of changes that can be handled by \DynFO and \DynFOLL programs, respectively.}
\label{table:overview}
\end{table}

Our results rely on two main techniques for handling changes of polylogarithmic size:
\begin{itemize}
 \item In the \emph{small-structure technique} (see Section \ref{section:small-structures}), it is exploited that on structures of polylogarithmic size, depth-$\bigO(\log\log n)$ circuits have the power of $\NC^2$ circuits. Dynamic programs that use this technique first construct a substructure of polylogarithmic size depending on the changes and the current auxiliary data, then perform a $\NC^2$-computation on this structure, and finally combine the result with the rest of the current auxiliary data to obtain the new auxiliary data. This technique is a slight generalization of previously used techniques for $\DynFO$.
 \item In the \emph{hierarchical technique} (see Section \ref{section:hierarchical}), it is exploited that auxiliary data used in dynamic programs is often ``composable''. Dynamic programs that use this technique first construct polynomially many structures depending on the current auxiliary data, each of them associated with one of the changes (in some cases, known dynamic programs for single changes can be exploited for this step). Then, in $\bigO(\log\log n)$ rounds, structures are combined hierarchically such that after $\ell$ rounds the program has computed polynomially many structures, each associated with $2^\ell$ changes. 
\end{itemize}

\section{Preliminaries and setting}\label{section:setting}
We introduce some notions of finite model theory, circuit complexity and the dynamic complexity framework.

\subparagraph*{Finite model theory \& circuit complexity.} A (relational) schema $\schema$ is a set of relation symbols and constant symbols. A relational structure $\struc$ over a schema $\schema$ consists of a finite domain~$\univ$, relations $R^\struc \subseteq \univ^k$ for every $k$-ary relation symbol $R \in \schema$, and interpretations $c^\struc \in \univ$ of every constant symbol $c \in \schema$. We assume in this work that every structure has a linear order $\leq$ on its domain. We can therefore identify $D$ with the set $\{0, \ldots, n-1\}$.

First-order logic \FO is defined in the usual way. Following \cite{ImmermanDC}, we allow first-order formulas to access the linear order on the structures and corresponding relations $+$ and $\times$ encoding addition and multiplication. We write $\FOar$ to make this explicit. \FOar can express iterated addition and iterated multiplication for polylogarithmically many numbers that consist of $\polylog(n)$ bits, see \cite[Theorem 5.1]{HAB}.

First-order logic with $\leq, +, \times$ is equivalent to (DLOGTIME) uniform $\AC^0$, the class of problems decidable by uniform families of constant-depth circuits with polynomially many ``not''-, ``and''- and ``or''-gates with unbounded fan-in. We write $\AC[f(n)]$ for the class that allows for polynomial-sized circuits of depth $\bigO(f(n))$, where $n$ is the number of input bits. For polynomially bounded and first-order constructible functions $f$, the class $\AC[f(n)]$ is equal to $\Ind[f(n)]$, the class of problems that can be expressed by inductively applying an \FOar formula $\bigO(f(n))$ times \cite[Theorem 5.22]{ImmermanDC}. So, we can think of an $\AC[f(n)]$ circuit as being a stack of $\bigO(f(n))$ copies of some $\ACz$ circuit.
The class \FOLL, see \cite{BarringtonKLM01}, is defined as $\Ind[\log \log n] = \AC[\log \log n]$.

The circuit complexity classes uniform $\NC^i$ and $\SAC^i$ are defined via uniform circuits of polynomial size and depth $\bigO((\log n)^i)$; besides ``not''-gates, $\NC$ circuits use ``and''- and ``or''-gates with fan-in $2$, $\SAC$ circuits allow for ``or''-gates with unbounded fan-in.

\subparagraph*{Dynamic complexity.}

The goal of a \emph{dynamic program} $\prog$ is to maintain the result of a query applied to an input structure $\inp$ that is subject to changes. 
In this paper, we consider changes of the form $\ins_R(P)$, the insertion of a set $P$ of tuples into the relation $R$ of $\inp$, and $\del_R(P)$, the deletion of the set $P$ from $R$. We usually restrict the size of the set $P$ to be bounded by a function $s(n)$, where $n$ is the size of the domain of $\inp$. Most of the time, the bound is polylogarithmic in $n$, so $s(n) = \log(n)^c$ for some constant $c$.
A pair $(\query, \Delta)$ of a query $\query$ and a set $\Delta$ of (size-bounded) change operation $\ins_R, \del_R$ is called a \emph{dynamic query}.

To maintain some dynamic query over $\schema$-structures, for some schema $\schema$, $\prog$ stores and updates a set $\calA$ of auxiliary relations over some schema $\auxSchema$ and over the same domain as the input structure.
For every auxiliary relation symbol $A \in \auxSchema$ and every change operation $\delta$, $\prog$ has an update program $\varphi_\delta^A(\tpl x)$, which can access input and auxiliary relations.
Whenever an input structure $\inp$ is changed by a change $\delta(P)$, resulting in the structure $\inp'$, the new auxiliary relation $A^{\aux'}$ in the updated auxiliary structure $\aux'$ consists of all tuples $\tpl a$ such that $\varphi_\delta^A(\tpl a)$ is satisfied in the structure $(\inp', \aux)$.

We say that a dynamic program $\prog$ \emph{maintains} a dynamic query $(\query, \Delta)$, if after applying a sequence $\alpha$ of changes over $\Delta$ to an initial structure $\inp_0$ and applying the corresponding update programs to $(\inp_0, \aux_0)$, where $\aux_0$ is an initial auxiliary structure, a dedicated auxiliary relation is always equal to the result of evaluating $\query$ on the current input structure.
Following Patnaik and Immerman \cite{PatnaikI97}, we demand that the initial input structure $\inp_0$ is \emph{empty}, so, has empty relations. The initial auxiliary structure is over the same domain as $\inp_0$ and is defined from $\inp_0$ by some first-order definable initialization.

The class \DynFO is the class of all dynamic queries that are maintained by a dynamic program with \FOar formulas as update programs\footnote{Other papers write \DynFO for the class that uses \FO update formulas without a priori access to the arithmetic relations $\leq, +, \times$ and \DynFOar for the class that uses \FOar update formulas. If changes only affect single tuples, there is no difference for most interesting queries, see \cite[Proposition~7]{DattaKMSZ18}. For changes that affect sets of tuples of non-constant size, all \DynFO maintainability results use \FOar update formulas, as \FO update formulas without arithmetic are not strong enough to maintain interesting queries. We therefore just write \DynFO and omit the suffix $(\leq, +, \times)$ to avoid visual clutter.}. Equivalently, we can think of the update programs as being $\ACz$ circuits.
The class $\DynFO[f(n)]$ allows for $\AC[f(n)]$ circuits as update programs. We often use the equivalence $\AC[f(n)] = \Ind[f(n)]$ and think of update programs that apply an \FOar update formula $f(n)$ times.
In this paper, we are particularly interested in the class $\DynFOLL = \DynFO[\log \log n]$.

\section{The small-structure technique}\label{section:small-structures}
The small-structure technique has been used for obtaining maintenance results for $\DynFO$ for non-constant size changes \cite{DattaKMTVZ20,TschirbsVZ23}. The idea is simple: for changes of size $m$, (1) compute a structure with a domain of size roughly $m$, depending on the changes and the current auxiliary data, then (2) compute information about this structure (as $m \ll n$, this computation can be more powerful than $\AC^0$), and (3) combine the result with the current auxiliary data to obtain the new auxiliary data. 

For $\DynFO$ and changes of polylogarithmic size, one can use $\SAC^1$-computations in step~(2), as formalized in the next lemma.
\begin{lemma}[{\cite[Corollary~3]{TschirbsVZ23}}]\label{theorem:smallstructure-fo}
Let $\query$ be a $k$-ary query on $\schema$-structures, for some $k \in \N$.
If $\query$ is uniform $\SAC^1$-computable, then there is an \FOar formula $\varphi$ over schema $\schema \cup \{C\}$ such that for any $\schema$-structure $\struc$ with $n$ elements, any subset $C$ of its domain of size $\polylog(n)$ and any $k$-tuple $\tpl a \in {C}^k$ it holds that: $\tpl a \in \query(\restrict{\struc}{C})$ if and only if $(\struc, C) \models \varphi(\tpl a)$. 
Here, $\restrict{\struc}{C}$ denotes the substructure of $\struc$ induced by~$C$.
\end{lemma}

For $\DynFOLL$, this generalizes in two directions: (a) for structures of size $\polylog(n)$ one can use $\NC^2$-computations, (b) for structures of size $(\log n)^{\bigO(\log \log n)}$ one can use $\SAC^1$-computations. This is captured by the following lemma.

\begin{lemma}\label{theorem:smallstructure-foll}
Let $\query$ be a $k$-ary query on $\schema$-structures, for some $k \in \N$.
\begin{enumerate}[(a)] 
 \item If $\query$ is uniform $\NC^2$-computable, then there is an \FOLL formula $\varphi$ over schema $\schema \cup \{C\}$ such that for any $\schema$-structure $\struc$ with $n$ elements, any subset $C$ of its domain of size $\polylog(n)$ and any $k$-tuple $\tpl a \in {C}^k$ it holds that: $\tpl a \in \query(\restrict{\struc}{C})$ if and only if $(\struc, C) \models \varphi(\tpl a)$. 
 \item If $\query$ is uniform $\SAC^1$-computable, then there is an \FOLL formula $\varphi$ over schema $\schema \cup \{C\}$ such that for any $\schema$-structure $\struc$ with $n$ elements, any subset $C$ of its domain of size $\superpolylog$ and any $k$-tuple $\tpl a \in {C}^k$ it holds that: $\tpl a \in \query(\restrict{\struc}{C})$ if and only if $(\struc, C) \models \varphi(\tpl a)$. 
\end{enumerate}

\end{lemma}
\begin{proof}%
\begin{enumerate}[(a)] 
\item Let $C$ have size $m$, which is polylogarithmically bounded in $n$. The $\NC^2$-circuit for $\query$ has polynomial size in $m$ and depth $\bigO((\log m)^2)$, so its size is polylogarithmic in $n$ and the depth is $\bigO((\log \log n)^2)$. It is well-known\footnote{Divide the circuit into layers of depth $\log \log n$. Each layer depends only on $\log n$ gates of the previous layer, as each gate has fan-in at most $2$, and can be replaced by a constant-depth circuit for the CNF of the layer, which has polynomial size.} that for every $\NC$-circuit of depth $f(n)$ there is an equivalent $\AC$-circuit of depth $\bigO(\frac{f(n)}{\log \log n})$ and size polynomial in the original circuit, so we can obtain an \AC-circuit for answering $\query$ on $C$ with depth $\bigO(\log \log n)$.
\item The proof of \cite[Lemma 8.1]{AllenderHMPS08} can easily be extended towards the following statement: if a language $L$ is decided by a non-deterministic Turing machine with polynomial time bound $m^c$ and polylogarithmic space bound $(\log m)^d$ then for every positive, non-decreasing and first-order constructible function $t(n)$ there is a uniform \AC circuit family for $L$ with depth $\bigO(t(n))$ and size $2^{\bigO(m^{\frac{c}{t(n)}} (\log m)^d)}$.
For $m = \log n^{e \log \log n}$ and $t(n) = 2ce \log \log n$, the size is exponential in $\sqrt{\log n} \, (\log \log n)^{\bigO(1)}$, and therefore, as this function grows slower than $\log n$, polynomial in $n$.
The statement follows as all $\SAC^1$ languages can be decided by a non-deterministic Turing machine with polynomial time bound and space bounded by $(\log n)^2$, see \cite{Cook70}.
\end{enumerate}%
\end{proof}

A straightforward application of the technique to dynamic programs from the literature yields the following \DynFOLL-programs. For directed reachability, adapting the \DynFO-program for $\bigO(\frac{\log n}{\log \log n})$ changes from \cite{Datta0VZ18} yields (with more proof details in the appendix):

\begin{theorem}\label{thm:reachMain}
Reachability in directed graphs is in \DynFOLL under insertions and deletions of $\polylog(n)$ edges.
\end{theorem}

For undirected reachability and regular languages, replacing Lemma~\ref{theorem:smallstructure-fo} by Lemma~\ref{theorem:smallstructure-foll}(b) in the \DynFO maintainability proofs for $\polylog(n)$ changes from \cite{DattaKMTVZ20, TschirbsVZ23} directly yields:

\begin{theorem}\label{thm:undirected}
\begin{enumerate}[(a)] 
 \item Reachability in undirected graphs is in \DynFOLL under insertions and deletions of $\superpolylog$ edges.
 \item Membership in regular languages is in \DynFOLL under symbol changes at $\superpolylog$ positions.
\end{enumerate}
\end{theorem}

The small-structure technique has further applications beyond graph reachability and regular languages. We mention a few here. The proofs are deferred to the appendix.

\begin{theorem}\label{thm:small-structure-further}
\begin{enumerate}[(a)] 
 \item A minimum spanning forest for weighted graphs can be maintained
  \begin{enumerate}[(i)]
  \item in \DynFO under changes of $\polylog(n)$ edges, and
  \item in \DynFOLL under changes of $\superpolylog$ edges.
\end{enumerate}   
 \item A maximal matching can be maintained in \DynFOLL under changes of $\polylog(n)$ edges.
 \item For graphs with maximum degree bounded by a constant $\delta$, a proper $(\delta+1)$-colouring can be maintained in $\DynFO$ under changes of $\polylog(n)$ edges.
\end{enumerate}
\end{theorem}

\section{The hierarchical technique}\label{section:hierarchical}
\newcommand{\dreieck}[1]{
	\draw (#1) -- ($(#1.center)+(-0.25,-0.5)$) -- ($(#1)+(0.25,-0.5)$)  -- (#1.center);
}
\newcommand{\caseOneA}{
\begin{tikzpicture}[
			yscale=1,
			xscale=1
		]
		
\node (tmp) at (-4,0){(1A)};
		
				\node[blackNode, label=right:$x$, draw] (x) at (0,1){};
				\node[blackNode, label=right:$r$] (r) at (0,-0){};
				\node[blackNode, label=right:$z$] (z) at (0,-1){};

				\node[blackNode, label=right:$h$] (h) at (0,-3){};
				\draw [snakeEdgea] (x) to (r);
				\draw [snakeEdgea] (r) to (z);
				\draw [snakeEdgea] (z) to (h);

				\node (D1) at (-1,-1.25){$D_1$};
				\node (D2) at (-2,-4){$D_2$};
				
				\begin{pgfonlayer}{background}
					\filldraw[fill=blue!10] (r) -- (-6, -4.5) -- (-5,-4.5) -- (z) -- (5,-4.5) -- (6,-4.5) -- (r);
					\filldraw[fill=blue!10] (z) -- (-5, -4.5) -- (-2,-4.5) -- (h) -- (2,-4.5) -- (5,-4.5) -- (z);
				\end{pgfonlayer}

	 	\end{tikzpicture}
	 }
 \newcommand{\caseTwoA}{
 	\begin{tikzpicture}[
 	yscale=1,
 	xscale=1
 	]
 	
 			\node (tmp) at (-5,0.5){(2A)};
 			
 				\node[blackNode, label=right:$x$, draw] (x) at (0,1){};
 \node[blackNode, label=right:$r$] (r) at (0,-0){};
 \node[blackNode, label=right:$z$] (z) at (0,-1){};

 \node[blackNode, label=right:$y$] (y1) at (-1,-2){};
 \node[blackNode, label=right:$ $] (y2) at (0,-2){};
 \node[blackNode, label=right:$ $] (y3) at (1,-2){};

 \dreieck{y2}
 \dreieck{y3}
 
 \draw [dotsEdge] (y2) to (y3);
 \draw [dEdge] (z) to (y1);
 \draw [dEdge] (z) to (y2);
 \draw [dEdge] (z) to (y3);

 \node[blackNode, label=right:$h$] (h) at (-1.5,-3){};
 \draw [snakeEdgea] (x) to (r);
 \draw [snakeEdgea] (r) to (z);
 \draw [snakeEdgea] (y1) to (h);
 
 \node (D1) at (-1,-1.25){$D_1$};
 \node (D2) at (-2.5,-4){$D_2$};
 
 \begin{pgfonlayer}{background}
 \filldraw[fill=blue!10] (r) -- (-6, -4.5) -- (-5,-4.5) -- (z) -- (5,-4.5) -- (6,-4.5) -- (r);
 \filldraw[fill=blue!10] (y1) -- (-4, -4.5) -- (-2,-4.5) -- (h) -- (-1,-4.5) -- (0,-4.5) -- (y1);
 \end{pgfonlayer}

 	\end{tikzpicture}
 }
\newcommand{\caseOneB}{
	\begin{tikzpicture}[
	yscale=1,
	xscale=0.85
	]
	\node (tmp) at (-5,0){(1B)};
	\node[blackNode, label=right:$x$, draw] (x) at (0,1){};
	\node[blackNode, label=right:$r$] (r) at (0,-0){};
	\node[blackNode, label=right:$v$] (v) at (0,-1){};
	
	\node[blackNode, label=right:$y_1$] (y1) at (-1.5,-2){};
	\node[blackNode, label=right:$y_2$] (y2) at (-0.5,-2){};
	\node[blackNode, label=right:$ $] (y3) at (0.5,-2){};
	\dreieck{y3}
	\node[blackNode, label=right:$ $] (y4) at (1.5,-2){};
	\dreieck{y4}
	
	\node[blackNode, label=right:$h$] (h) at (-2.5,-3){};
	\node[blackNode, label=right:$z$] (z) at (-0.5,-3){};

	\draw [snakeEdgea] (x) to (r);
	\draw [snakeEdgea] (r) to (v);
	\draw [snakeEdgea] (y1) to (h);
	\draw [snakeEdgea] (y2) to (z);
	
	\draw [dEdge] (v) to (y1);
	\draw [dEdge] (v) to (y2);
	\draw [dEdge] (v) to (y3);
	\draw [dEdge] (v) to (y4);

	\draw [dotsEdge] (y3) to (y4);

	\node (D1) at (-1,-1){$D_1$};
	\node (D2) at (-4,-4){$D_2$};
	\node (D3) at (-1.5,-4){$D_3$};
	\node (D4) at (-0.25,-4){$T$};
	
	\begin{pgfonlayer}{background}
	\filldraw[fill=blue!10] (r) -- (-8, -4.5) -- (-7,-4.5) -- (v) -- (7,-4.5) -- (8,-4.5) -- (r);
	\filldraw[fill=blue!10] (y1) -- (-6, -4.5) -- (-4,-4.5) -- (h) -- (-3.5,-4.5) -- (-2.5,-4.5) -- (y1);
	\filldraw[fill=blue!10] (y2) -- (-2.25, -4.5) -- (-1.5,-4.5) -- (z) -- (1,-4.5) -- (2,-4.5) -- (y2);
	\filldraw[fill=blue!10] (z) -- (-1.5, -4.5)  -- (1,-4.5) -- (z);
	\end{pgfonlayer}
	
	\end{tikzpicture}
}
\newcommand{\caseTwoB}{
	\begin{tikzpicture}[
	yscale=1,
	xscale=0.85
	]
	 			\node (tmp) at (-6,0.5){(2B)};
 	
	\node[blackNode, label=right:$x$, draw] (x) at (0,1){};
	\node[blackNode, label=right:$r$] (r) at (0,-0){};
	\node[blackNode, label=right:$v$] (v) at (0,-1){};
	
	\node[blackNode, label=right:$y_1$] (y1) at (-1.5,-2){};
	\node[blackNode, label=right:$y_2$] (y2) at (-0.5,-2){};
	\node[blackNode, label=right:$ $] (y3) at (0.5,-2){};
	\dreieck{y3}
	\node[blackNode, label=right:$ $] (y4) at (1.5,-2){};
	\dreieck{y4}
	
	\node[blackNode, label=right:$h$] (h) at (-2.5,-3){};
	\node[blackNode, label=right:$z$] (z) at (-0.5,-3){};
	
	\node[blackNode, label=right:$ $] (z1) at (-0.75,-4){};
	\node[blackNode, label=right:$ $] (z2) at (0.25,-4){};
	\dreieck{z1}
	\dreieck{z2}
	
	\draw [snakeEdgea] (x) to (r);
	\draw [snakeEdgea] (r) to (v);
	\draw [snakeEdgea] (y1) to (h);
	\draw [snakeEdgea] (y2) to (z);
	
	\draw [dEdge] (v) to (y1);
	\draw [dEdge] (v) to (y2);
	\draw [dEdge] (v) to (y3);
	\draw [dEdge] (v) to (y4);
	
	\draw [dEdge] (z) to (z1);
	\draw [dEdge] (z) to (z2);
	
	\draw [dotsEdge] (y3) to (y4);
	\draw [dotsEdge] (z1) to (z2);
	
	\node (D1) at (-1,-1){$D_1$};
	\node (D2) at (-4,-4){$D_2$};
	\node (D3) at (-1.5,-4){$D_3$};
	
	\begin{pgfonlayer}{background}
	\filldraw[fill=blue!10] (r) -- (-8, -4.5) -- (-7,-4.5) -- (v) -- (7,-4.5) -- (8,-4.5) -- (r);
	\filldraw[fill=blue!10] (y1) -- (-6, -4.5) -- (-4,-4.5) -- (h) -- (-3.5,-4.5) -- (-2.5,-4.5) -- (y1);
	\filldraw[fill=blue!10] (y2) -- (-2.25, -4.5) -- (-1.25,-4.5) -- (z) -- (1,-4.5) -- (2,-4.5) -- (y2);
	\end{pgfonlayer}
	
	\end{tikzpicture}
}
In this section we describe and use a simple, yet powerful hierarchical technique for handling polylogarithmic changes in $\DynFOLL$. After changing $m \df (\log n)^c$ many tuples, auxiliary data $\calR^1, \ldots, \calR^k$ is built in $k \df d \log \log n$ rounds, for suitable $d$. The auxiliary data $\calR^{\ell - 1}$ after round $\ell-1$ encodes information for certain subsets of the changes of size $2^{\ell-1}$. This information is then combined, via first-order formulas, to information on $2^\ell$ changes in round~$\ell$. The challenge for each concrete dynamic query is to find suitable auxiliary data which is defined depending on a current instance as well as on subsets of changes, and can be combined via first-order formulas to yield auxiliary data for larger subsets of changes.

We apply this approach to maintaining distances in acyclic and undirected graphs, context-free language membership, and tree-isomorphism under polylogarithmic changes. In these applications of the hierarchical technique, information is combined along paths, binary trees, and arbitrary trees, respectively.

\subsection{Undirected and acyclic reachability and distances}

The articles that introduced the class \DynFO showed that reachability for undirected and for acyclic graphs is in \DynFO under single-edge changes \cite{DongS93, PatnaikI97}. For these classes of graphs,  also distances, that is, the number of edges in a shortest path between two reachable nodes, can be maintained. For undirected graphs, this was proven in \cite{GradelS12}, for acyclic graphs it is a straightforward extension of the proof for reachability from \cite{DongS93}. 

While reachability for undirected graphs is in \DynFO under polylogarithmically many edge changes \cite{DattaKMTVZ20}, we only know the general $\bigO(\frac{\log n}{\log \log n})$ bound for acyclic graphs \cite{Datta0VZ18}. It is unknown whether distances can be maintained in \DynFO under changes of non-constant size, both for undirected and for directed, acyclic graphs.

\begin{theorem}\label{thm:distances}
Distances can be maintained in \DynFOLL under insertions and deletions of $\polylog(n)$ edges for 
\begin{enumerate}[(a)]
 \item undirected graphs, 
 \item acyclic directed graphs.
\end{enumerate}
\end{theorem}

To maintain distances, a dynamic programs can use a relation of the form $\dist(u,v,d)$ with the meaning ``the shortest path from $u$ to $v$ has length $d$''. The proof of Theorem~\ref{thm:distances} is then a direct application of the hierarchical technique on paths. After inserting polylogarithmically many edges, distance information for two path fragments can be iteratively combined to distance information for paths fragments that involve more changed edges. Thus  $\polylog{n}$ path fragments (coming from so many connected components before the insertion) can be combined in $\log \log n$ many iterations. 

To handle edge deletions, we observe that some distance information is still guaranteed to be valid after the deletion: the shortest path from $u$ to $v$ surely has still length $d$ after the deletion of some edge $e$ if there was no path of length $d$ from $u$ to $v$ that used $e$. These ``safe'' distances can be identified using the \dist relation. We show that after deleting polylogarithmically many edges, shortest paths can be constructed from polylogarithmically many ``safe'' shortest paths of the original graph. We make this formal now.

\begin{lemma}
	\label{lem:shortestpaths}
Let $G = (V,E)$ be an undirected or acyclic graph, $e\in E$ an edge and $u,v\in V$ nodes such that there is a path from $u$ to $v$ in $G'=(V,E-e)$.
 For every shortest path $u = w_0, w_1, \ldots, w_{d-1}, w_d=v$ from $u$ to $v$ in $G'$ there is an edge $(w_i, w_{i+1})$ such that no shortest path from $u$ to $w_i$ and no shortest path from $w_{i+1}$ to $v$ in the original graph $G$ uses~$e$.
\end{lemma} 
\begin{proof}
For undirected graphs, this was proven in \cite[Lemma 3.5c]{pang2005incremental}. We give the similar proof for acyclic graphs.
If no node $w_{i+1}$ on a shortest path $u = w_0, w_1, \ldots, w_{d-1}, w_d=v$ from $u$ to $v$ in $G'$ exists such that some shortest path from $u$ to $w_{i+1}$ in $G$ uses the edge $e$, the edge $(w_{d-1},v)$ satisfies the lemma statement. Otherwise, let $w_{i+1}$ be the first such node on the path. It holds $i+1 \geq 1$, as the shortest path from $u$ to $u$ trivially does not use $e$. So, no shortest path from $u$ to $w_i$ in $G$ uses $e$.
There is no shortest path from $w_{i+1}$ to $v$ in $G$ that uses $e$: otherwise, there would be a path from $e$ to $w_{i+1}$ and a path from $w_{i+1}$ to $e$ in $G$, contradicting the assumption that $G$ is acyclic.
\end{proof}

\begin{corollary}\label{cor:shortestpaths}
	Let $G = (V, E)$ be an undirected or acyclic graph and $\Delta E \subseteq E$ with $|\Delta E| = m$. For all nodes $u$ and $v$ such that $v$ is reachable from $u$ in $G' \df (V, E \setminus \Delta E)$ there is a shortest path in $G'$ from $u$ to $v$ that is composed of at most $m$ edges and $m+1$ shortest paths of $G$, each from some node $u_i$ to some node $v_i$ for $i \leq m+1$, such that no shortest path from $u_i$ to $v_i$ in $G$ uses an edge from $\Delta E$.
\end{corollary}
\begin{proof}[Proof idea.] Via induction over $m$.  For $m = 1$, this follows from \cref{lem:shortestpaths}. For $m > 1$ this is immediate from the induction hypothesis. 
 \end{proof}

We can now prove that distances in undirected and in acyclic graphs can be maintained in \DynFOLL under changes of polylogarithmic size.
 
\begin{proof}[Proof of \cref{thm:distances}]
We construct a $\DynFOLL$ program that maintains the auxiliary relation $\dist(u,v,d)$ with the meaning ``the shortest path from $u$ to $v$ has length $d$''.  

Suppose $m \df (\log n)^c$ edges $ \Delta E$ are changed in $G = (V, E)$ yielding the graph $G' = (V, E')$. W.l.o.g. all edges in $\Delta E$ are either inserted or deleted. In both cases, the program executes a first-order initialization, yielding auxiliary relations $\dist^0(u,v,d)$, and afterwards executes a first-order procedure for $k \df c \log \log n$ rounds, yielding auxiliary  relations $\dist^1, \ldots, \dist^k$. The superscripts on the relations are for convenience, they are all subsequently stored in $\dist$.

For insertions, we use the standard inductive definition of reachability and distances. Set $\dist^0(u,v,d) \df \dist(u,v,d)$, where $\dist(u,v,d)$ is the distance information of the unchanged graph $G$. Then, for $k$ rounds, the distance information is combined with the new edges $\Delta E$, doubling the amount of used edges from $\Delta E$ in each round. Thus $\dist^{\ell}(u,v,d)$ is computed from $\dist^{\ell-1}$ by including $\dist^{\ell-1}$ and all tuples which satisfy the formula:
\[\varphi_\ins \df \exists z_1\exists z_2\exists d_1 \exists d_2 ( 
\Delta E(z_1,z_2) \land d_1+d_2+1 = d \land \dist^{\ell-1}(u,z_1,d_1) \land \dist^{\ell-1}(z_2,v,d_2)
)\]
 
For deletions, the program starts from shortest paths $u, \ldots, v$ in $G$ such that no shortest path from $u$ to $v$ uses edges from $\Delta E$ and then combines them for $k$ rounds, which yields the correct distance information for $G'$ according to \cref{cor:shortestpaths}. Thus, the first-order initialization yields $\dist^0(u,v,d)$ via
\[\dist(u,v,d) \land \lnot \exists z_1\exists z_2\exists d_1\exists d_2 
( d=d_1+d_2+1 \land \Delta E(z_1,z_2) \land \dist(u,z_1,d_1) \land \dist(z_2,v,d_2) )\]

Then $\dist^{\ell}(u,v,d)$ is computed from $\dist^{\ell-1}(u,v,d)$ via a formula
similar to $\varphi_\ins$, using $E$ instead of $\Delta E$. 
\end{proof}

\subsection{Context-free language membership}

Membership problems for formal languages have been studied in dynamic complexity starting with the work of Gelade, Marquardt, and Schwentick \cite{GeladeMS12}. It is known that context-free languages can be maintained in \DynFO under single symbol changes \cite{GeladeMS12} and that regular languages can even be maintained under polylog changes \cite{Vortmeier19, TschirbsVZ23}.

It is an open problem whether membership in a context-free language can be maintained in \DynFO for changes of non-constant size. We show that this problem is in $\DynFOLL$ under changes of polylogarithmic size.

\begin{theorem}
	\label{thm:cfl_dynfoll}
	 Every context-free language can be maintained in \DynFOLL under changes of size $\polylog n$.
\end{theorem}

Suppose $G=(V,\Sigma,S,\Gamma)$ is a grammar in Chomsky normal form with $L \df  L(G)$. 
For single changes, 4-ary auxiliary relations $R_{X\to Y}$ are used for all $X,Y\in V$ \cite{GeladeMS12},  with the intention that $(i_1,j_1,j_2,i_2)\in R_{X\to Y}$ iff $X\Rightarrow^* w[i_1,j_1) Y w(j_2,i_2]$, where $w \df w_1 \ldots w_n$ is the current string. Let us call $I = (i_1,j_1,j_2,i_2)$ a \emph{gapped interval}. For a gapped interval $I$ and a set $P$ of changed positions, denote by $\#(I, P)$ the number of changed positions $p \in P$ with $p \in [i_1,j_1) \cup (j_2,i_2]$. 

The idea for the $\DynFOLL$ program for handling polylog changes is simple and builds on top of the program for single changes. It uses the same auxiliary relations and, after changing a set $P$ of positions, it collects gapped intervals $I$ into the relations $R_{X\to Y}$ for increasing $\#(I, P)$ in at most $\bigO(\log \log n)$ rounds. Initially, gapped intervals with $\#(I, P) \leq 1$ are collected using the first-order update formulas for single changes. Afterwards, in each round, gapped intervals $I$ with larger $\#(I, P)$ are identified by splitting $I$ into two gapped intervals $I_1$ and $I_2$ with  $\#(I, P) =  \#(I_1, P) +  \#(I_2, P)$ such that $I$ can be constructed from $I_1$ and $I_2$ with a first-order formula.

To ensure that $\bigO(\log \log n)$ many rounds suffice, we need that the intervals $I_1$ and $I_2$ can always be chosen such that $\#(I_1, P)$ and $\#(I_2, P)$ are of similar size. This will be achieved via the following simple lemma, which will be applied to parse trees. For a binary tree $T = (V, E)$ with red coloured nodes $R \subseteq V$, denote by $\#(T, R)$ the number of red nodes of $T$.
For a tree $T$ and a node $v$, let $T_v$ be the subtree of $T$ rooted at $v$.

\begin{lemma}
	\label{la:treesplitting}
	For all rooted binary trees $T = (V, E, r)$ with red coloured nodes $R \subseteq V$, there is a node $v\in V$ such that:
	\begin{itemize}
		\item $\#(T_v, R) \le \frac{2}{3}\cdot\#(T, R)$ and %
		\item $\#(T \setminus T_v, R) \le \frac{2}{3}\cdot\#(T, R)$ %
	\end{itemize}
\end{lemma}
\begin{proof}[Proof idea.] Walk down the tree starting from its root by always choosing the child whose subtree contains more red coloured nodes. Stop as soon as the conditions are satisfied. 
\end{proof}

We now provide the detailed proof of Theorem \ref{thm:cfl_dynfoll}.

\begin{proof}[Proof (of Theorem \ref{thm:cfl_dynfoll})]
We construct a $\DynFOLL$ program that maintains the auxiliary relations $R_{X\to Y}$ for all $X,Y\in V$. Suppose $m \df (\log n)^c$ positions $P$ are changed. The program executes a first-order initialization, yielding auxiliary relations $R^0_{X \to Y}$, and afterwards executes a first-order procedure for $k \df d \log \log n$ rounds, for $d \in \N$ chosen such that $(\frac{3}{2})^k > m$, yielding auxiliary  relations $R^1_{X \to Y}, \ldots, R^k_{X \to Y}$. The superscripts on the relations are for convenience, they are all subsequently stored in $R_{X\to Y}$.

For initialization, the $\DynFOLL$ program includes gapped intervals $(i_1, j_1, j_2, i_2)$ into the  relations $R^0_{X\to Y}$ for which
\begin{itemize}
 \item no position in $[i_1, j_1) \cup (j_2, i_2]$ has changed and $(i_1, j_1, j_2, i_2)$ was previously in $R_{X\to Y}$, or
 \item exactly one position in $[i_1, j_1) \cup (j_2, i_2]$ has changed and the dynamic program for single changes from \cite{GeladeMS12} includes the tuple $(i_1, j_1, j_2, i_2)$ into $R_{X\to Y}$.
\end{itemize}

Afterwards, for $k$ rounds, the $\DynFOLL$ program applies the following first-order definable procedure to its auxiliary relations. A gapped interval $I = (i_1, j_1, j_2, i_2)$ is included into $R^\ell_{X\to Y}$ in round $\ell$ if it was included in $R^{\ell-1}_{X\to Y}$ or one of the following conditions hold  (see \cref{fig:cfl} for an illustration):
\begin{itemize}
	\item[(a)] There are gapped intervals $I_1 = (i_1, u_1, u_2, j_2)$ and $I_2 = (u_1, j_1, j_2, u_2)$  and a non-terminal $Z \in V$ such that $I_1 \in R^{\ell-1}_{X \to Z}$ and $I_2 \in R^{\ell-1}_{Z \to Y}$. 
 	This can be phrased as first-order formula as follows:
 	\begin{multline*}
 		\varphi_a \df \exists u_1,u_2 \Bigg[ \left(i_1\le u_1\le j_1\le j_2\le u_2\le i_2\right) \land \\
 		\bigvee_{Z\in V} \Big(R_{X\to Z}(i_1,u_1,u_2,i_2) \land R_{Z\to Y}(u_1,j_1,j_2,u_2)\Big)\Bigg]
 	\end{multline*}

	\begin{figure}[t]
		\centering

		\scalebox{0.7}{
		  \begin{tikzpicture}[
      xscale=0.7,
      yscale=0.7,
      ]
      \node (tmp) at (-3,6){(a)};
            
      \node (X) at (-0,6){$X$};
      \node (Zstar) at (-0,4){$Z$};
      \node (Y) at (-0,2){$Y$};

      \node (i1) at (-3,0)[label={below:$i_1$}]{};
      \node (u1) at (-2,0)[label={below:$u_1$}]{};
      \node (j1) at (-1,0)[label={below:$j_1$}]{};
      \node (j2) at (1,0)[label={below:$j_2$}]{};
      \node (u2) at (2,0)[label={below:$u_2$}]{};
      \node (i2) at (3,0)[label={below:$i_2$}]{};

      \draw [uEdge, shorten >=0pt] (X) to (i1.center);
      \draw [uEdge, shorten >=0pt] (X) to (i2.center);
      \draw [uEdge, shorten >=0pt] (Zstar) to (u1.center);
      \draw [uEdge, shorten >=0pt] (Zstar) to (u2.center);
      \draw [uEdge, shorten >=0pt] (Y) to (j1.center);
      \draw [uEdge, shorten >=0pt] (Y) to (j2.center);

      \draw [snakeEdgea] (X) to (Zstar);
      \draw [snakeEdgea] (Zstar) to (Y);
      
      \draw [uEdge, shorten >=0pt, shorten <=0pt] (i1.center) to (u1.center);
      \draw [uEdge, shorten >=0pt, shorten <=0pt] (u1.center) to (j1.center);
      \draw [uEdge, shorten >=0pt, shorten <=0pt] (j2.center) to (u2.center);
      \draw [uEdge, shorten >=0pt, shorten <=0pt] (u2.center) to (i2.center);
    \end{tikzpicture}
    }
		\scalebox{0.7}{
		  \begin{tikzpicture}[
      xscale=0.7,
      yscale=0.7,
      ]
            
      \node (tmp) at (-3,6){(b)};

      \node (X) at (0,6){$X$};
      \node (Z) at (0,4.5){$Z$};
      \node (Z1) at (-1,3){$Z_1$};
      \node (Z2) at (1,3){$Z_2$};
      \node (Zstar) at (-1.3,1.5){$Z'$};
      \node (Y) at (1.3,1.5){$Y$};

      \node (i1) at (-4,0)[label={below:$i_1$}]{};
      \node (v1) at (-3,0)[label={below:$v_1$}]{};
      \node (u1) at (-2,0)[label={below:$u_1$}]{};
      \node (u2) at (-1,0)[label={below:$u_2$}]{};
      \node (v2) at (0,0)[label={below:$v_2$}]{};
      \node (j1) at (1,0)[label={below:$j_1$}]{};
      \node (j2) at (2,0)[label={below:$j_2$}]{};
      \node (v3) at (3,0)[label={below:$v_3$}]{};
      \node (i2) at (4,0)[label={below:$i_2$}]{};

      \draw [snakeEdgea] (X) to (Z);
      \draw [snakeEdgea] (Z1) to (Zstar);
      \draw [snakeEdgea] (Z2) to (Y);
      \draw [dEdge, shorten >=0pt] (Z) to (Z1);
      \draw [dEdge, shorten >=0pt] (Z) to (Z2);
      
      \draw [uEdge, shorten >=0pt] (X) to (i1.center);
      \draw [uEdge, shorten >=0pt] (X) to (i2.center);
      \draw [uEdge, shorten >=0pt] (Z1) to (v1.center);
      \draw [uEdge, shorten >=0pt] (Z1) to (v2.center);
      \draw [uEdge, shorten >=0pt] (Z2) to (v2.center);
      \draw [uEdge, shorten >=0pt] (Z2) to (v3.center);

      \draw [uEdge, shorten >=0pt] (Zstar) to (u1.center);
      \draw [uEdge, shorten >=0pt] (Zstar) to (u2.center);

			\draw [uEdge, shorten >=0pt] (Y) to (j1.center);
      \draw [uEdge, shorten >=0pt] (Y) to (j2.center);
     
      \draw [uEdge, shorten >=0pt, shorten <=0pt] (i1.center) to (v1.center);
      \draw [uEdge, shorten >=0pt, shorten <=0pt] (v1.center) to (u1.center);
      \draw [uEdge, shorten >=0pt, shorten <=0pt] (u1.center) to (u2.center);
      \draw [uEdge, shorten >=0pt, shorten <=0pt] (u2.center) to (v2.center);
      \draw [uEdge, shorten >=0pt, shorten <=0pt] (v2.center) to (j1.center);
      \draw [uEdge, shorten >=0pt, shorten <=0pt] (j2.center) to (v3.center);
      \draw [uEdge, shorten >=0pt, shorten <=0pt] (v3.center) to (i2.center);

    \end{tikzpicture}
    }
		\scalebox{0.7}{
		  \begin{tikzpicture}[
      xscale=0.7,
      yscale=0.7,
      ]
            
                  \node (tmp) at (-3,6){(c)};

      \node (X) at (0,6){$X$};
      \node (Z) at (0,4.5){$Z$};
      \node (Z1) at (-1,3){$Z_1$};
      \node (Z2) at (1,3){$Z_2$};
      \node (Zstar) at (1.3,1.5){$Z'$};
      \node (Y) at (-1.3,1.5){$Y$};

      \node (i1) at (-4,0)[label={below:$i_1$}]{};
      \node (v1) at (-3,0)[label={below:$v_1$}]{};
			\node (j1) at (-2,0)[label={below:$j_1$}]{};
      \node (j2) at (-1,0)[label={below:$j_2$}]{};
      \node (v2) at (0,0)[label={below:$v_2$}]{};
			\node (u1) at (1,0)[label={below:$u_1$}]{};
      \node (u2) at (2,0)[label={below:$u_2$}]{};
      \node (v3) at (3,0)[label={below:$v_3$}]{};
      \node (i2) at (4,0)[label={below:$i_2$}]{};

      \draw [snakeEdgea] (X) to (Z);
      \draw [snakeEdgea] (Z1) to (Y);
      \draw [snakeEdgea] (Z2) to (Zstar);
      \draw [dEdge, shorten >=0pt] (Z) to (Z1);
      \draw [dEdge, shorten >=0pt] (Z) to (Z2);
      
      \draw [uEdge, shorten >=0pt] (X) to (i1.center);
      \draw [uEdge, shorten >=0pt] (X) to (i2.center);
      \draw [uEdge, shorten >=0pt] (Z1) to (v1.center);
      \draw [uEdge, shorten >=0pt] (Z1) to (v2.center);
      \draw [uEdge, shorten >=0pt] (Z2) to (v2.center);
      \draw [uEdge, shorten >=0pt] (Z2) to (v3.center);

      \draw [uEdge, shorten >=0pt] (Zstar) to (u1.center);
      \draw [uEdge, shorten >=0pt] (Zstar) to (u2.center);

			\draw [uEdge, shorten >=0pt] (Y) to (j1.center);
      \draw [uEdge, shorten >=0pt] (Y) to (j2.center);
     
      \draw [uEdge, shorten >=0pt, shorten <=0pt] (i1.center) to (v1.center);
      \draw [uEdge, shorten >=0pt, shorten <=0pt] (v1.center) to (j1.center);

      \draw [uEdge, shorten >=0pt, shorten <=0pt] (j2.center) to (v2.center);
      \draw [uEdge, shorten >=0pt, shorten <=0pt] (v2.center) to (u1.center);
			\draw [uEdge, shorten >=0pt, shorten <=0pt] (u1.center) to (u2.center);
      \draw [uEdge, shorten >=0pt, shorten <=0pt] (u2.center) to (v3.center);
      \draw [uEdge, shorten >=0pt, shorten <=0pt] (v3.center) to (i2.center);

    \end{tikzpicture}
    }
		\caption{Illustration of when a gapped interval $(i_1, j_1, j_2, i_2)$ is added to $R_{X \to Y}$ in the proof of Theorem \ref{thm:cfl_dynfoll}.}
		\label{fig:cfl}
	\end{figure}
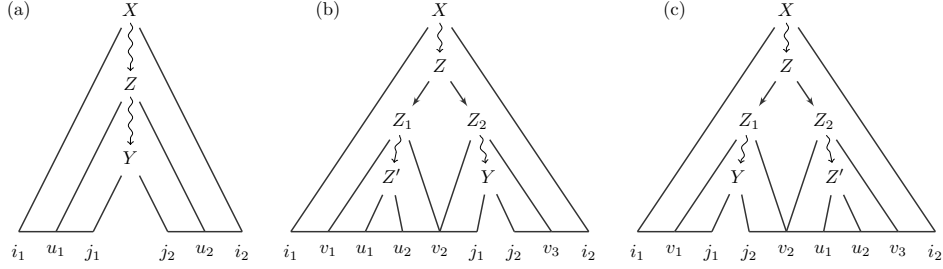

		\item[(b)] There are gapped intervals $I_1 = (i_1, v_1, v_3, i_2)$, $I_2 = (v_1, u_1, u_2, v_2)$, and $I_3 = (v_2, j_1, j_2, v_3)$ and non-terminals $Z, Z_1, Z_2, Z' \in V$ such that $Z \to Z_1Z_2 \in \Gamma$ and $I_1 \in R^{\ell-1}_{X \to Z}$, $I_2 \in R^{\ell-1}_{Z_1 \to Z'}$, $I_3 \in R^{\ell-1}_{Z_2 \to Y}$ and $w[u_1, u_2]$ can be derived from $Z'$.
	This can be phrased as first-order formula as follows:
	\begin{multline*}
		\varphi_b \df \exists u_1,u_2,v_1,v_2,v_3 \Bigg[\left( i_1\le v_1 \le u_1 \le u_2 \le v_2 \le j_1 \le j_2 \le v_3 \le i_2\right) \land  \\
		\bigvee_{\substack{Z,Z_1,Z_2,Z'\in V\\Z\to Z_1Z_2\in \Gamma}} \Big( R_{X\to Z}(i_1,v_1,v_3,i_2) \land R_{Z_1\to Z'}(v_1,u_1,u_2,v_2) \\ \land R_{Z'}(u_1,u_2) \land R_{Z_2\to Y}(v_2,j_1,j_2,v_3)  \Big)\Bigg]
	\end{multline*}
	Here, $R_{Z'}(u_1,u_2)$ is an abbreviation for the formula stating that $w[u_1, u_2]$ can be derived from $Z'$, i.e.
	$R_{Z'}(u_1,u_2) \df \exists v \bigvee_{W \to \sigma \in \Gamma} \left(R_{Z'\to W}(u_1, v, v, u_2) \land  \sigma(v)\right)$.
	
	\item[(c)] Symmetrical to (b), with gapped intervals $I_1 = (i_1, v_1, v_3, i_2)$, $I_2 = (v_1, j_1, j_2, v_2)$, and $I_3 = (v_2, u_1, u_2, v_3)$ and non-terminals $Z, Z_1, Z_2,Z' \in V$ such that $Z \to Z_1Z_2 \in \Gamma$ and $I_1 \in R^{\ell-1}_{X \to Z}$, $I_2 \in R^{\ell-1}_{Z_1 \to Y}$, $I_3 \in R^{\ell-1}_{Z_2 \to Z'}$ and $w[u_1, u_2]$ can be derived from $Z'$.

\end{itemize}

Note that $\#(I, P) = \#(I_1, P) + \#(I_2, P)$  in case (a) and $\#(I, P) = \#(I_1, P) + \#(I_2, P) + \#(I_3, P)$ in cases (b) and (c). Using \cref{la:treesplitting}, the intervals $I_j$ can be chosen such that $\#(I_j, P) \leq \frac{2}{3}\cdot\#(I, P)$ and thus $k$ rounds suffice. 

\end{proof}

It is known that for single tuple changes one can maintain for edge-labeled, acyclic graphs whether there is a path between two nodes with a label sequence from a fixed context-free language~\cite{MunozVZ16}. The techniques we have seen can be used to also lift this result to changes of polylogarithmic size.

\begin{proposition}
	Context-free path queries can be maintained under changes of polylogarithmic size in \DynFOLL on acyclic graphs.
\end{proposition}

\subsection{Tree isomorphism}

\newcommand{\subtree}[2]{\ensuremath{\text{subtree}_{#1}({#2})}}

The dynamic tree isomorphism problem -- given a forest $F = (V, E)$ and two nodes $x, x^* \in V$, are the subtrees rooted at $x$ and $x^*$ isomorphic? -- has been shown to be maintainable in $\DynFO$ under single edge insertions and deletions by Etessami \cite{Etessami98}.

It is not known whether tree isomorphism can be maintained in \DynFO under changes of
size $\omega(1)$. We show that it can be maintained in \DynFOLL under changes
of size $\polylog n$:

\begin{theorem}
	\label{thm:tiso:dynfoll}
	Tree isomorphism can be maintained in \DynFOLL under insertion and deletion of
	$\polylog n$ edges.
\end{theorem}
Intuitively, we want to use Etessami's dynamic program as the base case for a \DynFOLL-program: (1) compute isomorphism information for pairs of subtrees in which only one change happened, then (2) combine this information in $\log \log n$ many rounds. Denote by $\subtree{x}{r}$ the subtree rooted at $r$, in the tree rooted at $x$, within the forest $F$. The main ingredient in Etessami's program is a $4$-ary auxiliary relation $\textsc{t-iso}$ for storing tuples $(x,r, x^*, r^*)$ such that $\subtree{x}{r}$ and $\subtree{x^*}{r^*}$ are isomorphic and disjoint in $F$. It turns out that this information is not ``composable'' enough for step (2). 

We therefore slightly extend the maintained auxiliary information. A \emph{(rooted) context} $C = (V, E, r, h)$ is a tree $(V, E)$ with root $r \in V$ and one distinguished leaf $h \in V$, called the \emph{hole}. Two contexts $C = (V, E, r, h)$, $C^* = (V^*, E^*, r^*, h^*)$ are isomorphic if there is a root- and hole-preserving isomorphism between them, i.e. an isomorphism that maps $r$ to $r^*$ and $h$ to $h^*$. For a forest $F$ and nodes $x$, $r$, $h$ occurring in this order on some path, the context  $C(x, r, h)$ is defined as the context we obtain by taking $\subtree{x}{r}$, removing all children of $h$, and taking $r$ as root and $h$ as hole. 
Our dynamic program uses 
\begin{itemize}
 \item a $6$-ary auxiliary relation $\textsc{c-iso}$ for storing tuples $(C, C^*) \df (x, r, h, x^*, r^*, h^*)$ such that the contexts $C \df C(x, r, h)$ and by $C^* \df C(x^*, r^*, h^*)$ are disjoint and isomorphic,
 \item a ternary auxiliary relation $\textsc{dist}$ for storing tuples $(x, y, d)$ such that the distance between nodes $x$ and $y$ is $d$, and
  \item a $4$-ary auxiliary relation $\#\textsc{iso-siblings}$ for storing tuples $(x, r, y, m)$ such that $y$ has $m$ isomorphic siblings within $\subtree{x}{r}$.
\end{itemize}
The latter two relations have also been used by Etessami. From distances, a relation $\textsc{path}(x,y,z)$ with the meaning ``$y$ is on the unique path between $x$ and $z$'' is \FO-definable on forests, see \cite{Etessami98}. The relation $\textsc{t-iso}(x,r, x^*, r^*)$ can be $\FO$-defined from $\textsc{c-iso}$.

We will now implement the steps (1) and (2) with these adapted auxiliary relations. Suppose a forest $F \df (V, E)$ is changed into the forest $F' \df (V, E')$ by changing  a set $\Delta E$ of edges. A node $v \in V$ is \emph{affected} by the change, if $v$ is adjacent to some edge in $\Delta E$. The \DynFOLL program iteratively collects isomorphic contexts  $C$ and $C^*$ of $F'$ with more and more affected nodes. Denote by $\#(C, C^*, \Delta E)$ the number of nodes in contexts $C$ and $C^*$, excluding hole nodes, affected by  change $\Delta E$.

The following lemma states that \textsc{c-iso} can be updated for pairs of contexts with at most one affected node each. Its proof is very similar to Etessami's proof and is deferred to the appendix.

\begin{lemma}
	\label{la:tiso_single}
	Given $\textsc{c-iso}$, $\textsc{dist}$, and $\#\textsc{iso-siblings}$ and a set of changes $\Delta E$, the set of pairs $(C,C^\ast)$ of contexts  such that $C,C^\ast$
	are disjoint and isomorphic and such that both $C$ and $C^\ast$ contain at most one node affected by $\Delta E$ is \FO-definable.
\end{lemma}

 The dynamic program will update the auxiliary relation \textsc{c-iso} for contexts with at most one affected node per context using \cref{la:tiso_single}.  Isomorphic pairs $(C, C^*)$ of contexts with larger $\#(C, C^*, \Delta E)$ are identified by splitting both $C$ and $C^*$ into smaller contexts.

The splitting is done such that the smaller contexts have fewer than $\frac{2}{3} \cdot \#(C, C^*, \Delta E)$ affected nodes. To this end, we will use the following simple variation of \cref{la:treesplitting}.
For a tree $T=(V,E)$ and a function $p:V\to\{0,1,2\}$ which assigns each a node number of pebbles, let $\#(T,p)$ be the total number of pebbles assigned to nodes in $T$.

\begin{lemma}\label{lemma:splitting}
	Let $T$ be a tree of unbounded degree and $p$ such that either (i) $\#(T,p) > 2$, or (ii) $\#(T,p) = 2$ and $p(v)\le 1$ for all $v\in V$. Then there is a node $v$ such that either:
	\begin{itemize}
		\item[(1)] $\#(T\setminus T_v,p) \le \frac{2}{3} \cdot \#(T,p)$ and $\#(T_v,p) \le \frac{2}{3} \cdot \#(T,p)$, or
		\item[(2)] $\#(T\setminus T_v,p) \leq \frac{1}{3} \cdot \#(T,p)$ and $ \#(T_{u},p) \leq \frac{1}{3} \cdot \#(T,p)$ for any child $u$ of $v$.
	\end{itemize}
\end{lemma}

\begin{proof}[Proof idea] Use the same approach as for \cref{la:treesplitting}.
	If no node $v$ of type (1) is found, a node of type (2) must exist.
\end{proof}

We now prove that tree isomorphism can be maintained in \DynFOLL under changes of polylogarithmic size.

\begin{proof}[Proof (of \cref{thm:tiso:dynfoll})]
	 We construct a \DynFOLL program that maintains the auxiliary relations $\textsc{c-iso}$, $\textsc{dist}$, and $\#\textsc{iso-siblings}$. Suppose $m \df (\log n)^c$ edges $\Delta E$ are changed. As a preprocessing step, the auxiliary relation $\textsc{dist}$ is updated in depth $\bigO(\log \log n)$ via \cref{thm:distances}. Then, $\textsc{c-iso}$ is updated by first executing a first-order initialization for computing an initial version $\textsc{c-iso}^0$. Afterwards a first-order procedure is executed for $k \df d \log \log n$ rounds, for $d \in \N$ chosen such that $(\frac{3}{2})^k > m$, yielding auxiliary  relations $\{\textsc{c-iso}^\ell\}_{\ell \leq k}$ and $\{\#\textsc{iso-siblings}^\ell\}_{\ell \leq k}$. The superscripts on the relations are for convenience, they are all subsequently stored in $\textsc{c-iso}$, $\textsc{dist}$, and $\#\textsc{iso-siblings}$.
	 
	 The goal is that after the $\ell$th round
	 \begin{itemize}
	  \item $\textsc{c-iso}^\ell$ contains all pairs $C, C^*$ with \mbox{$\#(C, C^*, \Delta E) \leq (\frac{3}{2})^\ell$} which are isomorphic and disjoint, and
	  \item $\#\textsc{iso-siblings}^\ell$ contains the number of isomorphic siblings identified so far (i.e., with respect to $\textsc{c-iso}^\ell$).
	 \end{itemize}
		Round $\ell$ first computes $\textsc{c-iso}^\ell$ with a first-order procedure, and afterwards computes $\#\textsc{iso-siblings}^\ell$.
	For initialization, the $\DynFOLL$ program first computes $\textsc{c-iso}^0$, using \cref{la:tiso_single}, and $\#\textsc{iso-siblings}^0$.
	Afterwards, for $k$ rounds, the $\DynFOLL$ program combines known pairs of isomorphic contexts into pairs with more affected nodes and adapts $\textsc{c-iso}$ and $\#\textsc{iso-siblings}$ accordingly.
	
	\subparagraph*{Computing  $\textsc{C-ISO}^\ell$.} In the $\ell$th round, the program tests whether contexts $C \df C(x, r, h)$ and $C^* \df C(x^*, r^*, h^*)$ with $ \#(C, C^*, \Delta E) \leq (\frac{3}{2})^\ell$ are isomorphic by splitting both $C$ and $C^*$ into contexts with fewer affected nodes. The splitting is done by selecting suitable nodes $z \in C$ and $z^* \in C^*$, and splitting the context depending on these nodes.
	 
	 We first provide some intuition of how $z$ and $z^*$ are intended to be chosen. Suppose $C$ and $C^*$ are isomorphic via isomorphism $\pi$. With the goal of applying \cref{lemma:splitting},
	 let $p$ be the function that assigns to each non-hole node $v$ of $C$ a number of pebbles from $\{0, 1, 2\}$ indicating how many of the two nodes $v$ and $\pi(v)$ have been affected by the change $\Delta E$. Note that if $(C,p)$ does not fulfill the precondition of \cref{lemma:splitting}, then $(C,C^\ast)$ must have already been included in $\textsc{c-iso}^0$ during the initialization. Therefore, assume the precondition holds for $(C,p)$.
	 Let $C_z$ denote the subcontext of $C$ rooted at $z$.
	 Now, applying \cref{lemma:splitting} to $(C,p)$ yields a node $z$ such that one of the following cases holds:
	 \begin{itemize}
	 \item[(1)] $\#(C\setminus C_z, p) \leq \frac{2}{3} \cdot \#(C, p)$ and $\#(C_z, p) \leq \frac{2}{3} \cdot \#(C, p)$, or
	 \item[(2)] $\#(C\setminus C_z, p) \leq \frac{1}{3} \cdot \#(C, p)$ 
	 and $ \#(C_{u},p) \leq \frac{1}{3} \cdot \#(C, p)$ for any child $u$ of $z$.
	 \end{itemize}
 
	 Intuitively our first-order procedure tries to guess this node $z$ and its image $z^* \df \pi(z)$ and split the contexts $C$ and $C^*$ at these nodes. 

	 For testing that $C$ and $C^*$ are isomorphic, the program guesses two nodes $z$ and $z^*$ and (disjunctively) chooses case (1) or (2). Note that the program cannot be sure that it has correctly guessed $z$ and $z^*$ according to the above intuition. For this reason, the program first tests that the size restrictions from the chosen case (1) or (2) are fulfilled, which is easily possible in $\FOar$ as there are at most polylogarithmically many affected nodes. Since $\#(C,C^\ast,\Delta E)\le \left(\frac{3}{2}\right)^{\ell}$,
	 this ensures that, by induction, $\textsc{c-iso}^{\ell-1}$ is fully correct for all pairs of contexts that will be compared when testing isomorphism of $C$ and $C^\ast$.
	 Note that in case (2) the subtrees of any pair of children $u_1$ and $u_2$ of $z$ have at most $\frac{2}{3} \cdot \#(C,C^\ast,\Delta E)$ affected nodes. 	 
	 
	 Next, the procedure tests that there is an isomorphism between $C$ and $C^*$ that maps $z$ to $z^*$. The following claim is used:
	 
	 \begin{claim}\label{claim:node-iso}
	 	Suppose $z$ and $z^*$ are nodes in $C$ and $C^*$ that satisfy Condition (1) or (2) with children $Z \uplus Y$ and  $Z^* \uplus Y^*$, respectively, with $|Y| = |Y^*|$ constant. Then a first-order formula can test whether
	 	there is an isomorphism between the forests $\{\subtree{z}{u} \mid u \in Z\}$ and $\{\subtree{z^*}{u^*} \mid u^* \in Z^*\}$ using the relations $\textsc{c-iso}^{\ell-1}$ and $\#\textsc{iso-siblings}^{\ell-1}$.
	 \end{claim}
	 \begin{proof}
	 	The forests are isomorphic iff for each $u \in Z$ there is a $u^* \in Z^*$ such that $\subtree{z}{u} \cong \subtree{z^*}{u^*}$  and such that the number of nodes $v \in Z$ with $\subtree{z}{u} \cong \subtree{z}{v}$ is the same as the number of nodes $v^*$ with subtrees $\subtree{z^*}{u^*} \cong \subtree{z^*}{v^*}$, and vice versa with roles of $u$ and $u^*$ swapped.
	 	
	 	Because $z$ and $z^\ast$ satisfy condition (1) or (2), $\textsc{c-iso}^{\ell-1}$ is correct on the forest $\{\subtree{z}{u} \mid u \in Z\}\cup\{\subtree{z^*}{u^*} \mid u^* \in Z^*\}$ by induction. Additionally, $\#\textsc{iso-siblings}^{\ell-1}$ is consistent with $\textsc{c-iso}^{\ell-1}$ by induction.
	 	From $\textsc{c-iso}^{\ell-1}$, the tree isomorphism relation $\textsc{t-iso}^{\ell-1}$ -- storing tuples $(x,r, x^*, r^*)$ such that $\subtree{x}{r}$ and $\subtree{x^*}{r^*}$ are isomorphic and disjoint -- is \FO-definable.
	 	
	 	For testing whether a node $u\in Z$ satisfies the above condition, a first-order formula existentially quantifies a node $u^\ast\in Z^\ast$ and checks that $\textsc{t-iso}^{\ell-1}(z,u,z^\ast,u^\ast)$. The number of isomorphic siblings of $u,u^\ast$ is compared  using $\#\textsc{iso-siblings}^{\ell-1}$ and subtracting
	 	any siblings $y\in Y$ for which $\textsc{t-iso}^{\ell-1}(z,u,z,y)$ (and, respectively, $y^\ast\in Y^\ast$ for which $\textsc{t-iso}^{\ell-1}(z^\ast,u^\ast,z^\ast,y^\ast)$).
	 	This is possible in \FO because (a) $|Y|$ is constant and (b) $\#\textsc{iso-siblings}^{\ell-1}$ is consistent with $\textsc{c-iso}^{\ell-1}$ (even though $\textsc{c-iso}^{\ell-1}$ is not necessarily complete on subtrees in $Y,Y^\ast$).
	 \end{proof}
	 
	 For testing whether there is an isomorphism between $C$ and $C^*$ mapping $z$ to $z^*$, the program distinguishes the cases (1) and (2) from above. Further, in each of the cases it distinguishes (A) $\textsc{path}(r, z, h)$ and $\textsc{path}(r^*, z^*, h^*)$, or (B) $\neg \textsc{path}(r, z, h)$ and $\neg \textsc{path}(r^*, z^*, h^*)$. For all these first-order definable cases, a first-order formula can test whether there is an isomorphism between $C$ and $C^*$ mapping $z$ to $z^*$ as follows:\\
	 
	 \subparagraph{Case (1)}
	 Suppose $\#(C\setminus C_z, p) \leq \frac{2}{3} \cdot \#(C, p)$ and $\#(C_z, p) \leq \frac{2}{3} \cdot \#(C, p)$.

	 	\begin{itemize}
	 		\item[(A)]  Suppose $\textsc{path}(r, z, h)$ and $\textsc{path}(r^*, z^*, h^*)$. Then isomorphism of $C$ and $C^*$ is tested as follows  (see Figure~\ref{fig:tiso_1B}):
	 		\begin{itemize}
	 			\item[(a)] Test that the contexts $D_1 \df C(x, r, z)$ and $D^*_1 \df C(x^*, r^*, z^*)$ are isomorphic.
	 			\item[(b)] Test that the contexts $D_2 \df C(x, z, h)$ and $D^*_2 \df C(x^*, z^*, h^*)$ are isomorphic.
	 		\end{itemize}

	 		\item[(B)] Suppose $\neg\textsc{path}(r, z, h)$ and $\neg\textsc{path}(r^*, z^*, h^*)$. Then isomorphism of $C$ and $C^\ast$ is tested as follows  (see Figure~\ref{fig:tiso_1B}). Let $v$ be the least common ancestor of $z$ and $h$, and likewise $v^*$ the least common ancestor of $z^*$ and $h^*$. 
	 		\begin{itemize}
	 			\item[(a)] Test that the contexts $D_1 \df C(x,r,v)$ and $D^*_1 \df C(x^*, r^*, v^*)$ are isomorphic.
	 			\item[(b)] Let $y_1$ be the child of $v$ such that $\subtree{v}{y_1}$ contains the hole $h$ of $C$; and likewise for $y_1^*$ and $C^*$. Let $y_2$ be the child of $v$ such that $\subtree{v}{y_2}$ contains $z$; and likewise for $y_2^*$ and $z^*$. Then: 
	 			\begin{itemize}
	 				\item[(i)] Test that the contexts $D_2 \df C(v,y_1,h)$ and $D^*_2 \df C(v^*,y^*_1,h^*)$ are isomorphic.
	 				\item[(ii)] Test that the contexts $D_3 \df C(v, y_2, z)$ and $D^*_3 \df C(v^*, y^*_2, z^*)$ are isomorphic.
	 				\item[(iii)] Test that the trees $T \df \subtree{x}{z}$ and $T^* \df \subtree{x^*}{z^*}$ are isomorphic.
	 				\item[(iv)] Let $Z$ be the set of children of $v$ except for $y_1$, $y_2$; and likewise $Z^*$ for $v^\ast$ and $y_1^*, y_2^*$.
	 				Test that there is an isomorphism between the forests $\{\subtree{v}{u} \mid u \in Z\}$ and $\{\subtree{v^*}{u^*} \mid u^* \in Z^*\}$ (using Claim \ref{claim:node-iso} for $v$, $Z$, and $Y = \{y_1, y_2\}$).
	 			\end{itemize}
	 			
	 		\end{itemize}

	 	\end{itemize}

 \begin{figure}[t]
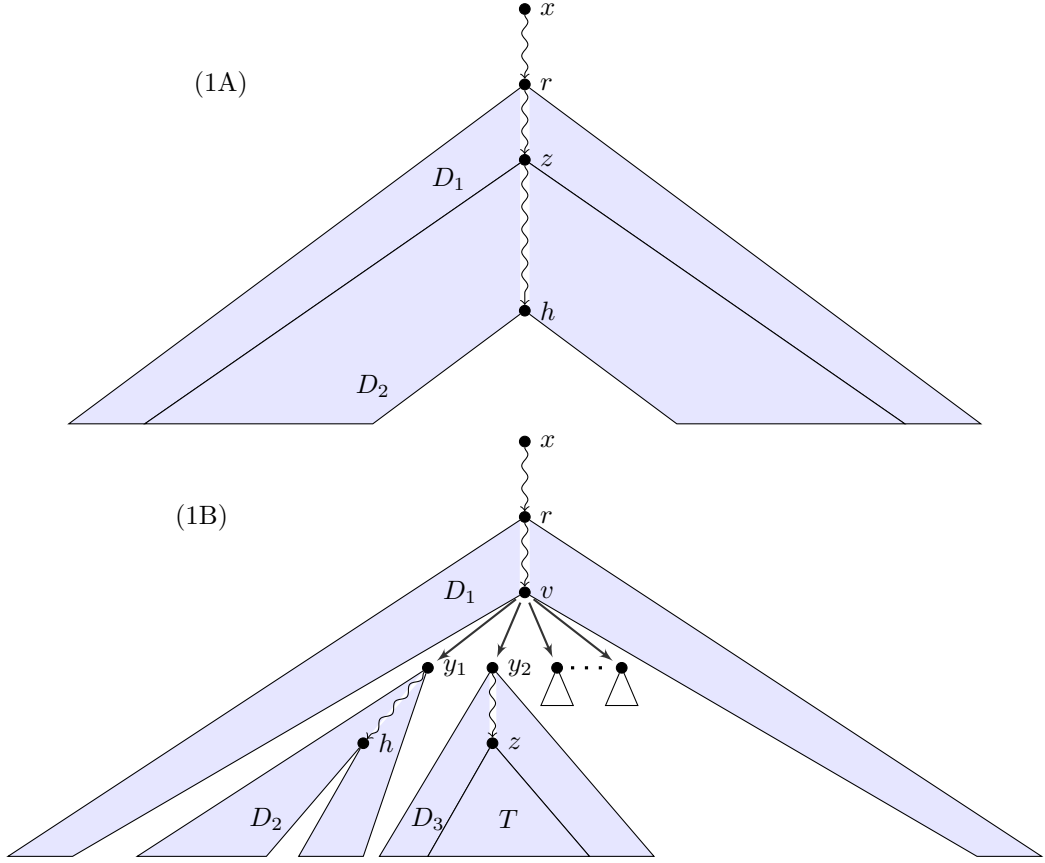

 	\centering
 	\caseOneA
 	\caseOneB
 	\caption{Illustration of the Cases (1A) and (1B) in the proof of \cref{thm:tiso:dynfoll}}
 	\label{fig:tiso_1B}
 \end{figure}
	 	
	 	\subparagraph{Case (2)}
	 	
	 	Let $\#(C\setminus C_z, p) \leq \frac{1}{3} \cdot \#(C, p)$  and $ \#(C_{u},p) \leq \frac{1}{3} \cdot \#(C, p)$ for any child $u$ of $z$.

	 	\begin{itemize}
	 		\item[(A)] Suppose $\textsc{path}(r, z, h)$ and $\textsc{path}(r^*, z^*, h^*)$. Then isomorphism of $C$ and $C^\ast$ is tested as follows (see Figure~\ref{fig:tiso_2A}):
	 		\begin{itemize}
	 			\item[(a)] Test that the contexts $D_1 \df C(x, r, z)$ and  $D^*_1 \df C(x^*, r^*, z^*)$ are isomorphic. 
	 			\item[(b)] Let $y$ be the child of $z$ such that $\subtree{z}{y}$ contains the hole $h$ of $C$; and likewise for $y^*$ and $C^*$. Let $Z$ be the set of children of $z$ except for $y$; and likewise $Z^*$ for $z^*$. Then:
	 			\begin{itemize}
	 				\item[(i)] Test that $D_2 \df C(z, y, h)$ and  $D^*_2 \df C(z^*, y^*, h^*)$ are isomorphic.
	 				\item[(ii)] Test that there is an isomorphism between the forests $\{\subtree{z}{u} \mid u \in Z\}$ and $\{\subtree{z^*}{u^*} \mid u^* \in Z^*\}$ (using Claim \ref{claim:node-iso} for $z$, $Z$, and $Y = \{y_1\}$).
	 				
	 			\end{itemize}
	 		\end{itemize}
	 		\item[(B)] Suppose $\neg\textsc{path}(r, z, h)$ and $\neg\textsc{path}(r^*, z^*, h^*)$. Then isomorphism of $C$ and $C^\ast$ is tested as follows (see Figure~\ref{fig:tiso_2A}). Let $v$ be the least common ancestor of $z$ and $h$, and likewise $v^*$ the least common ancestor of $z^*$ and $h^*$. Then:
	 		\begin{itemize}
	 			\item[(a)] Test that $D_1 \df C(x,r,v)$ and $D^*_1 \df C(x^*, r^*, v^*)$ are isomorphic.
	 			\item[(b)] Let $y_1$ be the child of $v$ such that $\subtree{v}{y_1}$ contains the hole $h$ of $C$; and likewise for $y_1^*$ and $C^*$. Let $y_2$ be the child of $v$ such that $\subtree{v}{y_2}$ contains $z$; and likewise for $y_2^*$ and $z^*$. Then:
	 			\begin{itemize}
	 				\item[(i)] Test that $D_2 \df C(v,y_1,h)$ and $D^*_2 \df C(v^*,y^*_1,h^*)$ are isomorphic.
	 				\item[(ii)] Test that $D_3 \df C(v, y_2, z)$ and $D^*_3 \df C(v^*, y^*_2, z^*)$ are isomorphic.
	 				\item[(iii)] Let $Z_1$ be the set of children of $v$ except for $y_1$, $y_2$; and likewise $Z_1^*$ for $v^\ast$ and $y_1^*, y_2^*$.
	 				Test that there is an isomorphism between the forests $\{\subtree{v}{u} \mid u \in Z_1\}$ and $\{\subtree{v^*}{u^*} \mid u^* \in Z_1^*\}$ (using Claim \ref{claim:node-iso} for $v$, $Z=Z_1$, and $Y = \{y_1, y_2\}$).
	 				\item[(iv)] Let $Z_2$ be the set of children of $z$; and likewise $Z_2^\ast$ for $z^\ast$. Test that there is an isomorphism between the forests $\{\subtree{z}{u} \mid u \in Z_2\}$ and $\{\subtree{z^*}{u^*} \mid u^* \in Z_2^*\}$ (using Claim \ref{claim:node-iso} for $z$, $Z=Z_2$, and $Y = \emptyset$).
	 			\end{itemize}
	 		\end{itemize}

	 	\end{itemize}

	 \begin{figure}[t]
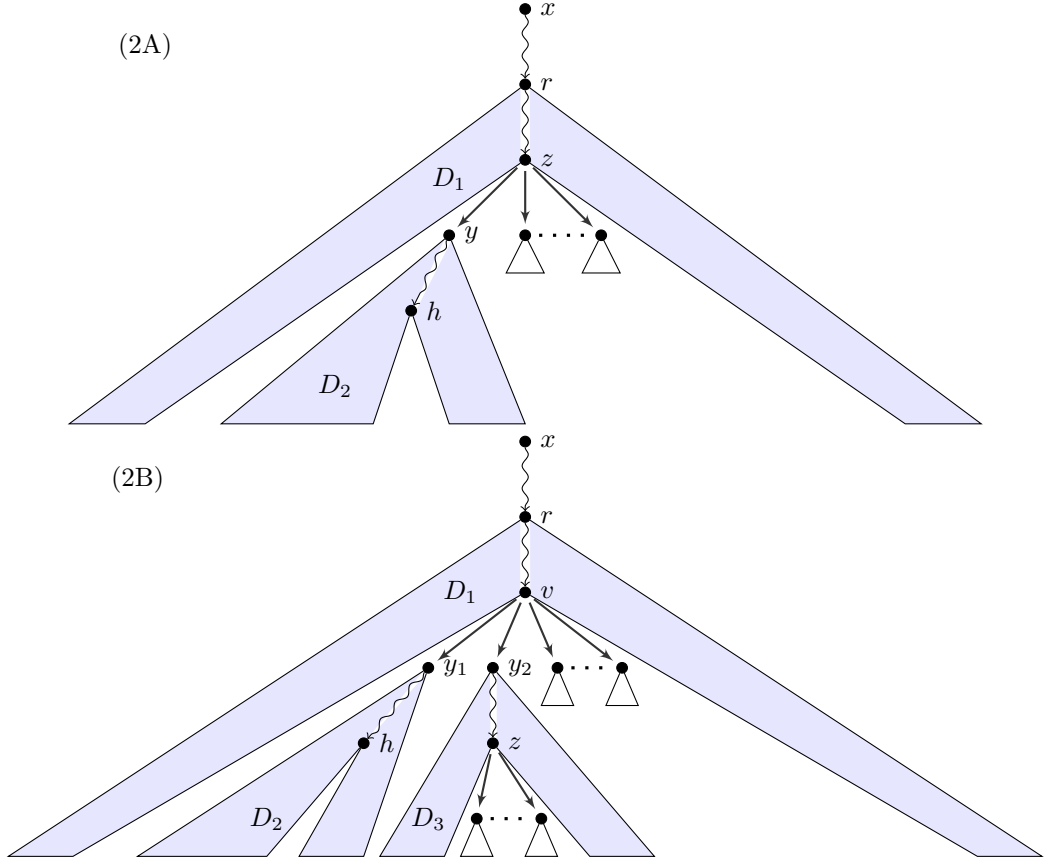

	 	\centering
	 	\caseTwoA
	 	\caseTwoB
	 	\caption{Illustration of the Cases (2A) and (2B) in the proof of \cref{thm:tiso:dynfoll}}
	 	\label{fig:tiso_2A}
	 \end{figure}
	 	
	 \newcommand{\siu}{\ensuremath{\#\textsc{iso-sibling}_{\text{unchanged}}}}
	
	\subparagraph*{Computing  $\#\textsc{ISO-SIBLINGS}^\ell$.}  The relation $\#\textsc{iso-siblings}^\ell$ can be first-order defined from $\textsc{c-iso}^{\ell}$ and the 4-ary relation $\siu$ containing tuples  
	 \begin{itemize}
	  \item $(x,r,y,m)$ with $m>0$ if $\subtree{x}{y}$ has no affected nodes and the number of isomorphic
	 siblings of $y$ in $\subtree{x}{r}$ with no affected nodes is $m$; and 
	 \item $(x,r,y,0)$ if \subtree{x}{y} contains an affected node.
	 \end{itemize}
	 Thus the relation $\siu$ contains isomorphism counts for ``unchanged'' siblings. It is \FO-definable from the old auxiliary data (from before the change) and the set of changes.

	 Given $\siu$ and $\textsc{c-iso}^\ell$, the relation
	 $\#\textsc{iso-sibling}^\ell$ is \FO-definable as follows.  Include a tuple $(x,r,y,m)$ into $\#\textsc{iso-sibling}^\ell$  if $m = m_1 + m_2$ where $m_1$ is the number of unchanged, isomorphic siblings of $y$ and $m_2$ is the number of isomorphic siblings of $y$ affected by the change (but by at most $(\frac{3}{2})^{\ell - 1}$ changes).  
	 
	 The number $m_1$ can be checked via distinguishing whether $\subtree{x}{y}$ has changed or not. 
	 If $\subtree{x}{y}$ has not changed, the formula checks that $m_1$ is such that $(x,r,y,m_1) \in \siu$. If $\subtree{x}{y}$ has changed, find a sibling $y^*$ of $y$ with an unchanged, isomorphic subtree. If $y^*$ exists then the formula checks that $m_1$ is such that $(x,r,y^*,m_1) \in \siu$, and otherwise that $m_1$ is 0. 
	 
	 For checking $m_2$, let $S(y)$ be the set of siblings $y^*$ of $y$ in \subtree{x}{r} that contain at least one affected node and where $\textsc{t-iso}^\ell(x,y,x,y^*)$. Since there are at most $\polylog n$ changes, $|S(y)|=\bigO(\polylog n)$. Therefore, $|S(y)|$ can be counted and compared to $m_2$ in \FO.
\end{proof}

\section{Tree decompositions of bounded-treewidth graphs}\label{section:bodlaender}
One of the best-known algorithmic meta-theorems is Courcelle's theorem, which states that all graph properties expressible in monadic second-order logic \MSO can be decided in linear time for graphs with tree-width bounded by some constant $k$ \cite{Courcelle90}. The tree-width is a graph parameter and measures how ``tree-like'' a graph is and is defined via tree decompositions, see below for details. Courcelle's theorem is based on Bodlaender's theorem, stating that in linear time (1) one can decide whether a graph has tree-width at most $k$ and (2) one can compute a corresponding tree decomposition \cite{Bodlaender96}.

Elberfeld, Jakoby and Tantau \cite{ElberfeldJT10} proved variants of these results and showed that ``linear time'' can be replaced with ``logarithmic space'' in both theorem statements. A dynamic version of Courcelle's theorem was proven in \cite{DattaMSVZ19}: every \MSO-definable graph property is in \DynFO under changes of single edges. The proof of the latter result circumvented providing a dynamic variant of Bodlaender's theorem, by using the result of Elberfeld et al.~that tree decompositions can be computed in \LOGSPACE, showing that a tree decomposition can be used to decide the graph property if only logarithmically many single-edge changes occurred after its construction, and that this is enough for maintenance in \DynFO.

It is an open problem to generalize the \DynFO maintenance result of \cite{DattaMSVZ19} from single-edge changes to changes of polylogarithmically many edges, even for \DynFOLL. Here, we provide an intermediate step and show that tree decompositions for graphs of bounded treewidth can be maintained in \DynFOLL. This result may lead to a second strategy for maintaining \MSO properties dynamically, in addition to the approach of \cite{DattaMSVZ19}.

A \emph{tree decomposition} $(T,B)$ of a graph $G = (V,E)$ consists of a rooted tree $T$ and a mapping $B$ from the nodes of $T$ to subsets of $V$. For a tree node $t$, we call the set $B(t)$ the \emph{bag} of $t$. 
A tree decomposition needs to satisfy three conditions. First, every vertex $v \in V$ needs to be included in some bag. Second, for every edge $(u,v) \in E$ there needs to be bag that includes both $u$ and $v$. Third, for each vertex $v \in V$, the nodes $t$ of $T$ such that $v \in B(t)$ form a connected subgraph in $T$.
The \emph{width} of a tree decomposition is the maximal size of a bag $B(t)$, over all tree nodes $t$, minus $1$. The \emph{treewidth} of a graph $G$ is the minimal width of a tree decomposition for $G$.

In addition to the width, important parameters of a tree decomposition are its \emph{depth}, the maximal distance from the root to a leaf, and its \emph{degree}, the degree of the tree $T$. Often, a binary tree decomposition of depth $\bigO(\log |V|)$ is desirable, while width $\bigO(k)$ for a graph of treewidth $k$ is tolerable. We show that one can maintain in \DynFOLL a tree decomposition of logarithmic depth but with unbounded degree. The proof does not use the hierarchical technique; a tree decomposition is defined in \FOLL from auxiliary information that is maintained in \DynFO.

\begin{theorem}\label{theorem:bodlaender}
For every $k$, there are numbers $c,d \in \N$ such that a tree decomposition of width $ck$ and depth $d \log n$ can be maintained in \DynFOLL under changes of $\polylog(n)$ edges for graphs of treewidth $k$, where $n$ is the size of the graph.
\end{theorem}

\begin{proof}
Elberfeld et al.~\cite[Section 3.2]{ElberfeldJT10} show that for any graph $G$ with $n$ vertices and treewidth at most $k$, one can compute in \LOGSPACE a tree decomposition of width $4k+3$ and depth at most $d \log n$, for some constant $d$ that only depends on $k$. We show that their algorithm can be adapted for \DynFOLL.

We summarize the approach of \cite{ElberfeldJT10} for constructing a tree decomposition of $G$.
As a first step, a so-called \emph{descriptor decomposition} is built \cite[Lemma 3.11, Lemma 3.13]{ElberfeldJT10}. 
A \emph{descriptor} $D$ is either a set $B$ of vertices of size at most $3k+3$, or a pair $(B, v)$ of a set $B$ of vertices of size at most $4k+4$ and a vertex $v \notin B$. Each descriptor defines a graph $G(D)$. For a descriptor $D = B$, the graph $G(D)$ is just $G[B]$, the induced subgraph of $B$ with vertex set $B$. 
For a descriptor $D = (B,v)$, let $C$ be the vertices of the connected component of $v$ in the graph $G \setminus B$, so, the graph that results from deleting $B$ from $G$. The graph $G(D)$ is the subgraph $G[B \cup C]$ of $G$.

We define a directed acyclic graph $M(G)$. The nodes of $M(G)$ are all possible descriptors. Edges go from parent descriptors to child descriptors as explained next, see \cite[Definition 3.8]{ElberfeldJT10}. 
A descriptor of the form $B$ has no children. 
For a descriptor of the form $D = (B,v)$, the goal is to find a separator of the graph $G(D)$, that is, a set $S$ of at most $k+1$ vertices such that $G(D) \setminus S$ consists of connected components that are as small as possible. Two cases are considered: if $B$ is large, that is, $B > 2k+2$, then $S$ is the lexicographically smallest set of at most $k+1$ vertices such that every connected component in $G(D) \setminus S$ contains at most half of the vertices from $B$. Otherwise, $S$ is chosen such that every connected component in $G(D) \setminus S$ contains at most half of the vertices from $G(D)$.
Let $C_1, \ldots, C_m$ be the components of $G(D) \setminus (B \cup S)$. For each $i \leq m$, the descriptor $(B_i,v_i)$ is a child of $(B,v)$ in $M(G)$ if $B_i$ is the subset of vertices from $B \cup S$ that have an edge to a vertex from $C_i$ and $v_i$ is the smallest vertex from $C_i$. 
Also, the descriptor $B \cup S$ is a child of $(B,v)$, unless $S \subseteq B$.

From $M(G)$, a tree decomposition is computed as follows. First, all parts of $M(G)$ that are not reachable from the descriptor $(\emptyset, v_0)$, where $v_0$ is some vertex, are discarded. The resulting structure is a tree. Each node corresponding to a descriptor $D=(B,v)$ is split into two nodes $D_n$ and $D_i$, with an edge from $D_n$ to $D_i$ and edges from $D_i$ to the children $D_1, \ldots, D_m$ of $D$ in $M(G)$. The node $D_n$ is labelled with the bag $B$, the node is labelled with the bag that contains all vertices that appear at least twice in $B, B_1, \ldots, B_m$, where the sets $B_i$ are the bags of the descriptors $D_i$.

The result is a tree decomposition of the claimed width and depth \cite[Lemma 3.12]{ElberfeldJT10}. 

We now argue how to adapt this approach for the dynamic setting. The main part is the definition of child descriptors in the construction of the graph $M(G)$. For this, we need to maintain the connected components of $G \setminus (B \cup S)$, where the sizes of $B$ and $S$ are constants that only depend on $k$. Note that this is possible in \DynFO under changes of polylogarithmically many edges.
\begin{claim}
Let $b$ be a constant. For any undirected graph $G$ that is subject to polylogarithmically many edge changes, one can maintain auxiliary information in \DynFO such that for any set $D$ of vertices of $G$ of size at most $b$, one can define the connected components of $G \setminus D$ in \FO.
\end{claim}
This claim follows from the fact that reachability in undirected graphs can be maintained in \DynFO under changes of polylogarithmically many edges \cite{DattaKMTVZ20}. For any subset $D$ of vertices with $|D|$ bounded by a constant, one can maintain reachability for the graph $G \setminus D$, as there are only polynomially such sets $D$.

As separators $S$ have constant size, we can quantify over them and select a lexicographically smallest one in \FO. Note that for a descriptor $D = (B,v)$ we need to select a separator that splits either the set $B$ or the vertices of $G(D)$ into small components, depending on the size of $B$. 
This is easily possible for splitting $B$, as this set has constant size and we can determine in \FO for each connected component of $G \setminus S$ the number of vertices from $B$ it includes. This is not possible for the vertices of $G(D)$. Nevertheless, for this case we can still identify a separator that is almost optimal, as one can approximately count in \FO.

\begin{claim}[{see \cite[Lemma 4.1]{ChaudhuriR96}}]
For every $\epsilon > 0$ there is an \FOar formula $\psi(x)$ such that for each structure $\calA$ with some unary relation $U$, there is an element $a$ such that $\psi(a)$ is satisfied in $\calA$, and for every $a$ with $\calA \models \psi(a)$, it holds that $(1- \epsilon)|U| \leq a \leq (1+\epsilon)|U|$.
\end{claim}

So, by the above claim and based on auxiliary information maintained in \DynFO, we can select in \FO a separator that splits the graph $G(D)$ into components of size at most $\frac{2}{3}$ of the original graph. This still guarantees that the graph constructed via the child descriptor relationship has logarithmic depth.

It remains to define the final tree decomposition from the constructed graph. The split of a descriptor node $D$ into two nodes $D_n$ and $D_i$ is obviously definable in \FO. We still need to discard all parts that are not reachable from the selected root descriptor $(\emptyset, v_0)$. As all directed paths in the graph have at most logarithmic length, reachability can be defined in \FOLL.

Altogether, this describes how to maintain a tree decomposition in \DynFOLL.
\end{proof}
 
\section{Conclusion and discussion}\label{section:discussion}
We have shown that most existing maintenance results for $\DynFO$ under single tuple changes can be lifted to $\DynFOLL$ for changes of polylogarithmic size. A notable exception are queries expressible in monadic second-order logic, which can be maintained on graphs of bounded treewidth under single-tuple changes.

Thus it seems very likely that one can find large classes of queries such that: If a query from the class can be maintained in \DynFO for changes of size $\bigO(1)$, then it can be maintained in $\DynFOLL$ for polylogarithmic changes. Identifying natural such classes of queries is an interesting question for future research.

\bibliography{bibliography}

\appendix

\section{Proof details for Section \ref{section:small-structures}}
We provide proof details for Theorem~\ref{thm:reachMain} and Theorem~\ref{thm:small-structure-further}.

\subsection{Directed Reachability}
In  \cite{Datta0VZ18} it was shown that reachability in directed graphs can be maintained in \DynFOar under changes of size $\bigO(\frac{\log{n}}{\log \log{n}})$.  We show:

\begin{theorem}[Restatement of Theorem~\ref{thm:reachMain}]\label{thm:reachAppendix}
    Reachability in directed graphs is in \DynFOLL under insertions and deletions of $\polylog(n)$ edges.
\end{theorem}
A straightforward adaption of the proof, using Lemma \ref{theorem:smallstructure-foll}(a), yields that reachability can be maintained under $\polylog(n)$ changes in \DynFOLL.

The proof strategy from \cite{Datta0VZ18} is as follows. First, reachability in a directed graph $G$ with $n$ nodes is reduced to the question whether an entry of a matrix inverse $A^{-1}$ is non-zero, for some $n \times n$ matrix $A$ that is related to the adjacency matrix of $G$. Edge changes that affect $k$ nodes translate to changing the matrix $A$ to $A + UBV$, for some matrices $U,B,V$, where $B$ has dimension $k \times k$, $U$ has dimension $n \times k$ and contains at most $k$ non-zero rows, and  $V$ has dimension $k \times n$ and contains at most $k$ non-zero columns. The update of the matrix inverse is computed using the Sherman-Morrison-Woodbury identity
\begin{equation*}
(A+UBV)^{-1} =A^{-1} -A^{-1}U(I+BVA^{-1}U)^{-1}BVA^{-1}.
\end{equation*}
Then it is argued that it is sufficient to evaluate the above expression modulo $\bigO(\log n)$ bit primes, so we can assume that all matrix entries are $\bigO(\log n)$ bit numbers. As the addition of $\polylog(n)$ many $\polylog(n)$ bit numbers is in \FOar, see \cite[Theorem 5.1]{HAB}, the matrix multiplications in the above expression can be expressed in \FOar, as long as $k$ is polylogarithmic in $n$. 
For evaluating the $k \times k$ matrix inverse $(I+BVA^{-1}U)^{-1}$, one only needs to be able to compute the determinant of a $k \times k$ matrix, as entries of matrix inverses can be expressed using the determinant of the matrix and of a submatrix. Thus the result follows from the following lemma.

\begin{lemma}[{\cite[Theorem 8]{Datta0VZ18}}]\label{lem:determinant}
Fix a domain of size $n$ and a $\bigO(\log n)$ bit prime $p$. The value of the 
determinant of a matrix $C \in \mathbb{Z}_p^{k\times k}$ for 
$k = \bigO(\frac{\log n}{\log \log n})$ can be defined in $\FOar$.
\end{lemma}

To prove Theorem \ref{thm:reachMain}, we only need to substitute Lemma~\ref{lem:determinant} by the following lemma from~\cite{ADKSVV23}.
\begin{lemma}[{\cite[Lemma~16]{ADKSVV23}}]\label{lem:ADKSVV}
Fix a domain of size $n$ and a $\bigO(\log n)$ bit prime $p$. The value of the 
determinant of a matrix $C \in \mathbb{Z}_p^{k\times k}$ for 
$k = \polylog(n)$ can be computed in $\FOLL$.
\end{lemma}
This lemma follows from the fact that the determinant of an $n \times n$ times integer matrix with $n$ bit entries can be computed in uniform $\TCo$, cf.~\cite{MV97}, and the fact that \FOLL can express all $\TCo \subseteq \NC^2$ queries on substructures of polylogarithmic size due to Lemma \ref{theorem:smallstructure-foll}(a). 

\subsection{Minimum Spanning Forest}
In this section, we consider the Minimum Spanning Forest (MSF) problem and show Theorem~\ref{thm:small-structure-further}(a). 
The goal is to maintain the edges of a minimum spanning forest of a given weighted undirected graph, where the edge weights are numbers represented by polynomially many bits with respect to the size of the graph. It is known that the problem is in $\DynFO$ under single edge changes~\cite{PatnaikI97}. To maintain reachability in undirected graphs, the dynamic programs maintain a spanning forest \cite{PatnaikI97, DattaKMTVZ20}, but not necessarily a minimum one.

\begin{lemma}[Restatement of Theorem~\ref{thm:small-structure-further}(a)]\label{thm:MST}
    A minimum spanning forest for weighted graphs can be maintained
  \begin{enumerate}[(i)]
  \item in \DynFO under changes of $\polylog(n)$ edges, and
  \item in \DynFOLL under changes of $\superpolylog$ edges.
\end{enumerate}   
\end{lemma}

Our approach is based on Prim's algorithm: for a graph $G = (V,E)$, we maintain an ordered list $e_1, \ldots, e_m$ of the edges of the graph, sorted in increasing order on the weight of the edges. For each index $i$, for $i \leq m$, we maintain the number of connected components of the graph $G_i$ over vertex set $V$ and edge set $E_i = \{e_1, \ldots, e_i\}$, so, the graph that contains the lightest $i$ edges. The minimum spanning forest of $G$ contains exactly those edges $e_i$ such that the graph $G_i$ has fewer connected components than $G_{i-1}$. 

\begin{proof}
We maintain a relation $L$ that contains a tuple $(i,u,v)$ if the edge $(u,v)$ is at position $i$ in the sorted list of all edges, where ties are broken arbitrarily. When a set $\Delta E$ of edges is inserted or deleted, we have to count for each position $i$ the number of edges that either were at a smaller position $j < i$ and were deleted, or were inserted and have a strictly smaller weight than the edge at position $i$. This offset determines the new position in the list of the edge that was at position $i$ before the change.
A newly inserted edge is inserted at position $j_1 + j_2 + 1$, where $j_1$ is the maximal position of an formerly present edge with weight smaller than the inserted edge and $j_2$ is the number of inserted edges that have a smaller weight or that have the same weight and are lexicographically smaller with respect to the linear order $\leq$ on the domain.
We observe that the necessary counting is restricted to a set of size $|\Delta E|$. As counting is in the complexity class $\TCz \subseteq \SAC^1$, it can be done on the restricted set in \FOar in case of polylogarithmically many edges changes (Lemma~\ref{theorem:smallstructure-fo}) and in \FOLL in case of $\superpolylog$ edge changes (Lemma~\ref{theorem:smallstructure-foll}(b)). Adding the numbers is possible in \FOar.

Let $G_i$ be the undirected graph that consists of the first $i$ edges of the list $L$, for each $i \leq |E|$. After each change $\Delta E$, the graph $G_i$ only changes by at most $|\Delta E|$ edges, so reachability in $G_i$ can be maintained in \DynFO and \DynFOLL, respectively, see \cite{DattaKMTVZ20} and Theorem~\ref{thm:undirected}(a). 
The edge $(u,v)$ at position $i$ in the list $L$ is part of the spanning forest if and only if the nodes $u$ and $v$ are in different connected components in the graph $G_{i-1}$.
\end{proof}

 \subsection{Maximal Matching}
In this section, we consider the Maximal Matching problem: the problem of finding a matching in a graph that is not a proper subset of another matching. Maximal Matching is in \DynFO under single-edge changes \cite{PatnaikI97}. We show here that a maximal matching can be maintained under $\polylog(n)$ changes in $\DynFOLL$. 

\begin{lemma}[Restatement of Theorem~\ref{thm:small-structure-further}(b)]\label{thm:maximal}
    A maximal matching can be maintained in \DynFOLL under changes of $\polylog(n)$ edges.
\end{lemma}

\begin{proof}
Let $G = (V,E)$ be the current graph and let $M \subseteq E$ be the maintained maximal matching, that is, a set of edges such that (1) each node $v \in V$ is incident to at most one edge from $M$ and (2) no further edge can be added to $M$ without violating this property. We call a node \emph{matched} if it is incident to some edge from $M$ and \emph{unmatched} otherwise.

When a set $\Delta E$ of polylogarithmically many edges is inserted into $G$, we only need to select a maximal subset of these edges that can be added to the maintained matching $M$. So, let $G_I$ be the graph that consists of all unmatched nodes of $G$ with respect to $M$ that are incident to some edge from $\Delta E$. 
As the size of $\Delta E$ is polylogarithmically bounded, so is the size of $G_I$. It follows from Lemma~\ref{theorem:smallstructure-foll}(a) and the fact that a maximal matching can be computed in $\NC^2$~\cite{Luby86} that we can compute a maximal matching $M_I$ of $G_I$ in \FOLL. The set $M \cup M_I$ is a maximal matching for the changed graph.

When a set $\Delta E$ of polylogarithmically many edges is deleted from $G$, we have to select a maximal set of existing edges that can replace deleted edges from $M$. Let $G'$ be the graph after the deletion.
We define a set $V_D$ of nodes as follows. First, let $V_1$ be the set of nodes that are incident to a deleted matching edge, so, to an edge in $M \cap \Delta E$. The number of these nodes is bounded by $(\log n)^c$, for some number $c$, as the size of $\Delta E$ is polylogarithmically bounded. For all nodes $v \in V_1$, we collect in the set $V_2$ all nodes $w$ such that $(v,w)$ is an edge in $G'$, $w$ is unmatched with respect to the matching $M \setminus \Delta E$ and among the $(\log n)^c$ first such neighbours of $v$ in $G'$ with respect to the linear order on the domain. We set $V_D = V_1 \cup V_2$.
The size of $V_D$ is bounded by $(\log n)^{2c}$, so, is polylogarithmic. Let $G_D$ be the subgraph of $G'$ that is induced by $V_D$.
Again, following Lemma~\ref{theorem:smallstructure-foll}(a), we can compute a maximal matching $M_D$ of $G_D$ in \FOLL. 

We argue that $M' = (M \setminus \Delta E) \cup M_D$ is a maximal matching of $G'$. It is clearly a matching, as all nodes of $G_D$ are unmatched with respect to $(M \setminus \Delta E)$. 
It is also maximal: suppose some edge $(u,v)$ could be added to $M'$ without violating the property of $M'$ being a matching. At least one of the two nodes is incident to an edge from $M \cap \Delta E$, as otherwise $M \cup \{(u,v)\}$ would be a matching of $G$, violating the assumption that $M$ was maximal. Suppose this node is $u$. So, $u$ is in $V_1 \subseteq V_D$. If also $v \in V_D$, then the edge $(u,v)$ is in $G_D$ and $M_D \cup \{(u,v)\}$ would be a matching of $G_D$, violating the assumption that $M_D$ is a maximal matching of $G_D$. If $v \not\in V_D$, that set contains $(\log n)^c$ other neighbours of $u$ that are unmatched with respect to $M \setminus \Delta E$. They cannot be all matched with other nodes by $M_D$, as they only have edges to nodes from $V_1$ (otherwise, those edges could be added to the original matching $M$) and there are only $(\log n)^c - 1$ such nodes. So, $M'$ would not be a maximal matching for $G_D$.
\end{proof}

\subsection{\texorpdfstring{$(\delta+1)$-colouring}{(δ+1)-colouring}}
A \emph{$c$-colouring} of a graph $G = (V,E)$ is a function $\text{col}: V\to C$, which assigns colours from a set $C$ with $|C| = c$ to vertices of $G$. Such a colouring is called \emph{proper} if $c(u)\neq c(v)$ for all $u, v \in V$ with $(u,v) \in E$,  i.e. no two neighbours are coloured the same. The set $C$ is called the \emph{palette} of the colouring. 

It is known that for a graph whose maximum degree is bounded by a constant $\delta$, a proper $(\delta+1)$-colouring can be computed in $\log^* n$ time on a $\mathsf{EREW}$ PRAM with $\bigO(n)$ many processors~\cite{GoldbergP87}. In particular, the problem can be solved in \FOLL.

In this section, we show that in the dynamic setting a proper $(\delta+1)$-colouring can be maintained in $\DynFO$ under changes of polylogarithmic size.  
\begin{lemma}[Restatement of Theorem~\ref{thm:small-structure-further}(c)]\label{lem:deltaColFO}
	For graphs with maximum degree bounded by a constant $\delta$, a proper $(\delta+1)$-colouring can be maintained in $\DynFO$ under changes of $\polylog(n)$ edges.
\end{lemma}

Suppose the current graph $G = (V, E)$ is proper $(\delta+1)$-coloured by $\text{col}$ and a set $\Delta$ of $\polylog(n)$ edges is changed. Let $A$ be the set of at most $\polylog(n)$ nodes affected by the change. The idea is to (1) compute a proper $(2^\delta+\delta+1)$-colouring for the graph, and then to (2) compute a proper $(\delta+1)$-colouring from the $(2^\delta+\delta+1)$-colouring.

For (1), the proper $(2^\delta+\delta+1)$-colouring is computed by composing a $2^\delta$-colouring for the induced subgraph on the affected vertices (with a disjoint palette from the existing one) with the stored $(\delta+1)$-colouring of the unaffected vertices. To compute the $2^\delta$-colouring of the affected nodes, we use that for graphs with constant degree bound $\delta$, a $2^\delta$ colouring can be computed in $\LOGSPACE$ \cite[Lemma 6]{DattaKM16}. A $2^\delta$ colouring on subgraphs induced by the at most polylogarithmically many affected nodes can thus be defined in $\FOar$ due to Lemma~\ref{theorem:smallstructure-fo}.

For (2), the construction of the $(\delta+1)$-colouring by a first-order formula (\cref{algo:deltaCol}) uses the construction of maximal independent sets from a $k$-colouring by a first-order formula as sub-routine (\cref{algo:mis}).

For constructing maximal independent set, the algorithm of Goldberg and Plotkin~\cite{GoldbergP87} for constructing such sets from a proper $k$-colouring for constant $k$ is used. Denote by $\mathcal{N}[I]$ the set of neighbours of a set $I$ of nodes.

\begin{algorithm}
	\caption{Constructing a maximal independent set for a graph $G = (V, E)$ with proper $k$-colouring provided as $\{X_1,X_2,\ldots,X_k\}$.}
	\label{algo:mis}
	\begin{algorithmic}[1]
		\State $X\gets V$, $I\gets \emptyset$, $\mathcal{X}=\{X_1,X_2,\ldots,X_k\}$.
		\While{$X\neq\emptyset$}
		\State $\text{Pick an arbitrary set }  X_l \text{ from }\mathcal{X}$.
		\State $I\gets I\cup X_l$ 
		\State $X\gets X\setminus\mathcal{N}[I]$
		\State $\mathcal{X} \gets \{X_i\setminus \mathcal{N}[I]| X_i\in\mathcal{X}, X_i\setminus \mathcal{N}[I]\neq \emptyset\}$
		\EndWhile
		\State \Return $I$
	\end{algorithmic}
\end{algorithm}

\begin{proposition}\label{prop:mis}
	\cref{algo:mis} terminates after at most $k$ iterations of the while loop, and the returned set $I$ is a maximal independent set of $G$. Moreover, the steps of the algorithm are $\FOar$ expressible.
\end{proposition}

\begin{algorithm}
	\caption{Constructing a proper colouring provided as $\calX = \{X_1,\ldots, X_{\delta+1}\}$ for a graph $G=(V,E)$ with degree bound $\delta$ and proper colouring $\calX' = \{X'_1,\ldots, X'_{k}\}$.}
	\label{algo:deltaCol}
	\begin{algorithmic}[1]
		\State $X\gets V$, $\mathcal{X}\gets\emptyset$ and $i\gets 1$.
		\While{$(X\neq \emptyset)$}\label{algo:deltaCol:loop}
		\State $X_i\gets$ maximal independent set of $G[X]$ via~\cref{algo:mis} \label{algo:deltaCol:mis}
		\State $X\gets X\setminus X_i$\label{algo:deltaCol:rem} 
		\State $\mathcal{X}\gets \mathcal{X}\cup \{X_i\}$
		\State $i\gets i+1$
		\EndWhile
		\State \Return $\mathcal{X}$
	\end{algorithmic}
\end{algorithm}
\begin{proposition}\label{prop:bCol}
	For a graph with degree bounded by a constant $\delta$ together with a proper colouring $\calX'$, \cref{algo:deltaCol} terminates after at most $\delta+1$ iterations of the while loop at line~\ref{algo:deltaCol:loop} and the returned $\mathcal{X}$ is a proper $(\delta+1)$-colouring. Moreover, the steps are $\FOar$ expressible.
\end{proposition}
\begin{proof}
The colouring $\calX'$ is always a proper colouring of the graph $G[X]$, so~\cref{algo:mis} can correctly compute a maximal independent set of $G[X]$. 

	To prove that the loop terminates after constantly many iterations, we do the following.
	We claim that for every vertex $v$ not in $X_i$ at line~\ref{algo:deltaCol:mis}, after the line~\ref{algo:deltaCol:rem}, the degree of $v$ in $G[X]$ is down by at least one. This is because if the degree of $v$ did not reduce after taking away the maximal independent set $X_i$ from $G[X]$ then it would not have had any edge to $X_i$, contradiction to the maximality of $X_i$. Since the maximum degree is $\delta$, the loop terminates after $\delta+1$ iterations, and $\calX$ contains at most $\delta+1$ sets.
	
	The colouring implicit in $\mathcal{X}$ is proper, as each set $X_i \in \calX$ is an independent set.
\end{proof}
\cref{prop:mis,prop:bCol} together prove~\cref{lem:deltaColFO}.

\section{Proof details for Section \ref{section:hierarchical}}

\subsection{Proof details for Tree isomorphism}
We provide proof details for Lemma~\ref{la:tiso_single}.

\begin{lemma}[Restatement of Lemma~\ref{la:tiso_single}]
	Given $\textsc{c-iso}$, $\textsc{dist}$, and $\#\textsc{iso-siblings}$ and a set of changes $\Delta E$, the set of pairs $(C,C^\ast)$ of contexts  such that $C,C^\ast$
	are disjoint and isomorphic and such that both $C$ and $C^\ast$ contain at most one node affected by $\Delta E$ is \FO-definable.
\end{lemma}

\begin{proof}
	
	We provide an \FO formula for collecting pairs of disjoint and isomorphic contexts with at most one affected node each. The idea for the formula is very similar to the update formulas for \textsc{t-iso} under single changes used by Etessami \cite{Etessami98}.

	Assume that a set $\Delta E$ of edges is inserted or deleted in $F$.
	Let $C=C(x,r,h)$ and $C^\ast=C(x^\ast,r^\ast,h^\ast)$ be two contexts such that $C$ and $C^\ast$ are disjoint and isomorphic, and each of them contains at most one node affected by $\Delta E$.
	
	The first-order formula does a case distinction on the total number of affected nodes:
	
	\begin{itemize}
		\item[(a)] Suppose $\#(C,C^\ast,\Delta E) = 0$, i.e. neither of the contexts is affected by the change. Then the old context isomorphism auxiliary data $\textsc{c-iso}$ for $(C,C^\ast)$ remains valid.
		\item[(b)] Suppose $\#(C,C^\ast,\Delta E) = 1$, i.e. exactly one of the contexts contains an affected node. W.l.o.g. node $a$ in $C$ is affected. For defining whether $(C,C^\ast)$ is in $\textsc{c-iso}$, the first-order formula distinguishes whether (1) $a$ is on the path from $r$ to $h$, i.e. $\textsc{path}(r,a,h)$, or (2) $a$ is not on that path, i.e. $\lnot\textsc{path}(r,a,h)$:
		\begin{itemize}
			\item[(1)] Suppose $\textsc{path}(r,a,h)$ (see Figure \ref{fig:tiso_single_onpath}).  If $C$ and $C^*$ are isomorphic, then there must be a node $a^\ast$ in $C^\ast$
			which is the image of $a$ under the presumed isomorphism.
			Let  $y$ and $y^\ast$ be children of $a$ and~$a^\ast$, respectively, such that
			$\textsc{path}(a,y,h)$ and $\textsc{path}(a^\ast,y^\ast,h^\ast)$.
			The following first-order definable tests verify that $C$ and $C^*$ are isomorphic:
			\begin{itemize}
			 \item[(i)] Test that $D_1 \df C(x,r,a)$ and $D^*_1\df C(x^*,r^*,a^*)$ are isomorphic.
			 \item[(ii)] Test that  $D_2 \df C(x,y,h)$ and $D^*_2 \df C(x^*,y^*,h^*)$ are isomorphic. 
			 \item[(iii)] Let $Z$ be the children of $a$ except for $y$; and likewise $Z^\ast$ for
			 $a^\ast$ and $y^\ast$. Test that there is an isomorphism between the forests $\{\subtree{a}{u}\mid u\in Z \}$ and $\{\subtree{a^\ast}{u^\ast}\mid u^\ast\in Z^\ast \}$ (using \cref{claim:node-iso} for $a$, $Z$ and $Y\df\{y\}$)
			\end{itemize}

			\begin{figure}[h]
				\centering
				\begin{tikzpicture}[
				yscale=1,
				xscale=1
				]
				\node[blackNode, label=right:$x$, draw] (x) at (0,1){};
				\node[blackNode, label=right:$r$] (r) at (0,-0){};
				\node[blackNode, label=right:$a$] (z) at (0,-1){};

				\node[blackNode, label=right:$y$] (y1) at (-1,-2){};
				\node[blackNode, label=right:$ $] (y2) at (0,-2){};
				\node[blackNode, label=right:$ $] (y3) at (1,-2){};

				\dreieck{y2}
				\dreieck{y3}
				
				\draw [dotsEdge] (y2) to (y3);
				\draw [dEdge] (z) to (y1);
				\draw [dEdge] (z) to (y2);
				\draw [dEdge] (z) to (y3);

				\node[blackNode, label=right:$h$] (h) at (-1.5,-3){};
				\draw [snakeEdgea] (x) to (r);
				\draw [snakeEdgea] (r) to (z);
				\draw [snakeEdgea] (y1) to (h);
				
				\node (D1) at (-1,-1.25){$D_1$};
				\node (D2) at (-2.5,-4){$D_2$};
				
				\begin{pgfonlayer}{background}
				\filldraw[fill=blue!10] (r) -- (-6, -4.5) -- (-5,-4.5) -- (z) -- (5,-4.5) -- (6,-4.5) -- (r);
				\filldraw[fill=blue!10] (y1) -- (-4, -4.5) -- (-2,-4.5) -- (h) -- (-1,-4.5) -- (0,-4.5) -- (y1);
				\end{pgfonlayer}

				\end{tikzpicture}
			
				\caption{Case (1): The node $a$ affected by $\Delta E$ is on the path between $r$ and $h$.}
				\label{fig:tiso_single_onpath}
			\end{figure}
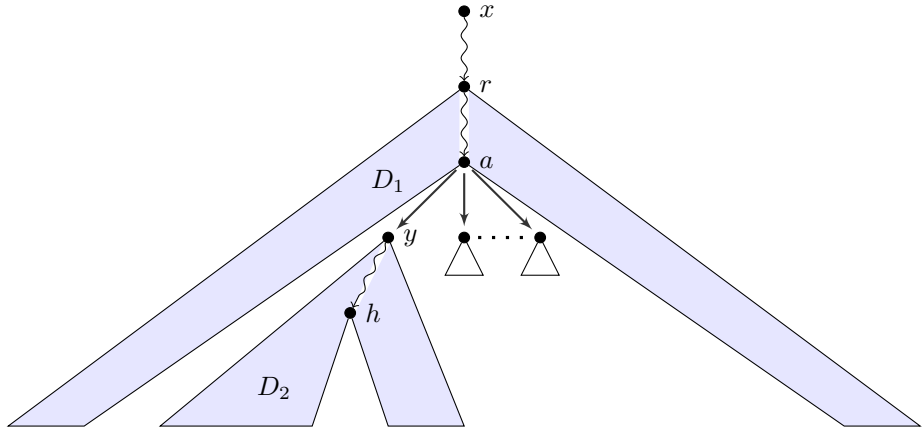

			\item[(2)] Suppose $\lnot\textsc{path}(r,a,h)$ (see \cref{fig:tiso_single_Lca}).
			If $C$ and $C^*$ are isomorphic, then there must be a node $a^\ast$ in $C^\ast$
			which is the image of $a$ under the presumed isomorphism. Let $z$ be the least common ancestor of $a$ and $h$, and $z^\ast$ the least common ancestor of $a^\ast$ and $h^\ast$, respectively. Let $y_1$ and $y_2$  be the children of $z$  such that $\textsc{path}(z,y_1,h)$ and $\textsc{path}(z,y_2,a)$; and likewise let $y_1^\ast$ and  $y_2^\ast$ be the children of $z^\ast$ such that $\textsc{path}(z^\ast,y_1^\ast,h^\ast)$ and $\textsc{path}(z^\ast,y_2^\ast,a^\ast)$. The following first-order definable tests verify that $C$ and $C^*$ are isomorphic:
			\begin{itemize}
				\item[(i)] Test that $D_1 \df C(x,r,z)$ and $D^*_1\df C(x^*,r^*,z^*)$ are isomorphic,
				\item[(ii)] Test that $D_2 \df C(x,y_1,h)$ and $D^*_2 \df C(x^*,y_1^*,h^*)$ are isomorphic,
				\item[(iii)] Test that $D_3 \df C(x,y_2,a)$ and $D^*_3 \df C(x^*,y_2^*,a^*)$ are isomorphic,
				\item[(iv)] Let $Z_1$ be the children of $z$ except $y_1,y_2$; and likewise $Z_1^\ast$ for
				$z^\ast$ and $y_1^\ast,y_2^\ast$.
				Test that there is an isomorphism between the forests $\{\subtree{z}{u}\mid u\in Z_1 \}$ and $\{\subtree{z^\ast}{u^\ast}\mid u^\ast\in Z_1^\ast \}$ (using \cref{claim:node-iso} for $z$, $Z\df Z_1$ and $Y\df\{y_1,y_2\}$)
				\item[(v)] Let $Z_2$ be the children of $a$; and likewise $Z_2^\ast$ for
				$a^\ast$. Test that there is an isomorphism between the forests $\{\subtree{a}{u}\mid u\in Z_2 \}$ and $\{\subtree{a^\ast}{u^\ast}\mid u^\ast\in Z_2^\ast \}$ (using \cref{claim:node-iso} for $a$, $Z\df Z_2$ and $Y\df\emptyset$)
				
			\end{itemize}
		
			\begin{figure}
			\centering
			\begin{tikzpicture}[
			yscale=1,
			xscale=0.85
			]
			\node[blackNode, label=right:$x$, draw] (x) at (0,1){};
			\node[blackNode, label=right:$r$] (r) at (0,-0){};
			\node[blackNode, label=right:$z$] (v) at (0,-1){};
			
			\node[blackNode, label=right:$y_1$] (y1) at (-1.5,-2){};
			\node[blackNode, label=right:$y_2$] (y2) at (-0.5,-2){};
			\node[blackNode, label=right:$ $] (y3) at (0.5,-2){};
			\dreieck{y3}
			\node[blackNode, label=right:$ $] (y4) at (1.5,-2){};
			\dreieck{y4}
			
			\node[blackNode, label=right:$h$] (h) at (-2.5,-3){};
			\node[blackNode, label=right:$a$] (z) at (-0.5,-3){};
			
			\node[blackNode, label=right:$ $] (z1) at (-1,-4){};
			\node[blackNode, label=right:$ $] (z2) at (0.0,-4){};
			\dreieck{z1}
			\dreieck{z2}
			
			\draw [snakeEdgea] (x) to (r);
			\draw [snakeEdgea] (r) to (v);
			\draw [snakeEdgea] (y1) to (h);
			\draw [snakeEdgea] (y2) to (z);
			
			\draw [dEdge] (v) to (y1);
			\draw [dEdge] (v) to (y2);
			\draw [dEdge] (v) to (y3);
			\draw [dEdge] (v) to (y4);
			
			\draw [dEdge] (z) to (z1);
			\draw [dEdge] (z) to (z2);
			
			\draw [dotsEdge] (y3) to (y4);
			\draw [dotsEdge] (z1) to (z2);
			
			\node (D1) at (-1,-1){$D_1$};
			\node (D2) at (-4,-4){$D_2$};
			\node (D3) at (-1.5,-4){$D_3$};
			
			\begin{pgfonlayer}{background}
			\filldraw[fill=blue!10] (r) -- (-8, -4.5) -- (-7,-4.5) -- (v) -- (7,-4.5) -- (8,-4.5) -- (r);
			\filldraw[fill=blue!10] (y1) -- (-6, -4.5) -- (-4,-4.5) -- (h) -- (-3.5,-4.5) -- (-2.5,-4.5) -- (y1);
			\filldraw[fill=blue!10] (y2) -- (-2.25, -4.5) -- (-1.5,-4.5) -- (z) -- (1,-4.5) -- (2,-4.5) -- (y2);
			\end{pgfonlayer}
			
			\end{tikzpicture}
				\caption{Case (2): The node $a$ affected by $\Delta E$ is not on the path between $r$ and $h$.}
				\label{fig:tiso_single_Lca}
			\end{figure}
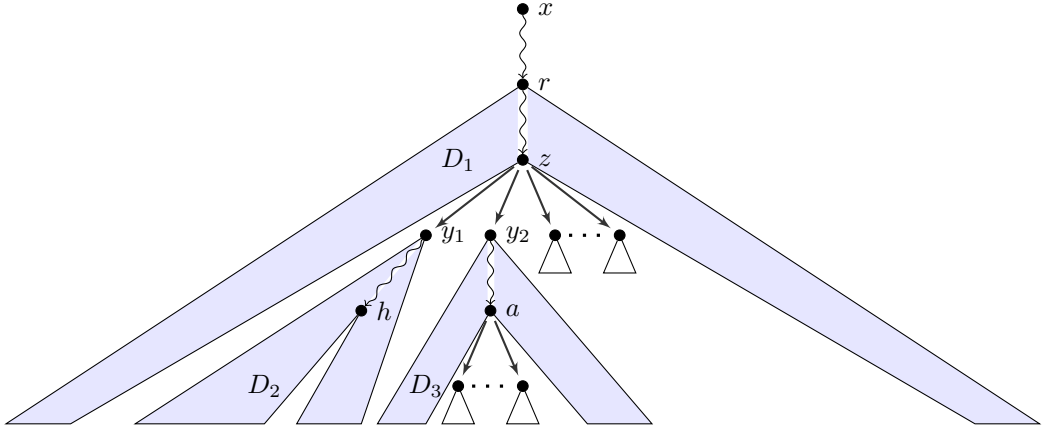
			
		\end{itemize}
		
		\item Suppose $\#(C,C^\ast,\Delta E) = 2$, i.e. both contexts contain exactly one affected node. Let $a$ be the affected node of $C$ and $b^*$ be the affected node of $C^\ast$. If $C$ and $C^\ast$ are isomorphic, then there must be nodes $a^\ast$ in $C^*$ and $b$ in $C$ that match $a$ and $b$ under the isomorphism. The test for verifying that $C$ and $C^*$ is similar to the case $\#(C,C^\ast,\Delta E) = 1$: An \FO formula splits $C$ at the nodes $a$, $b$, the hole $h$ and their respective least common ancestors, and checks that the unchanged contexts in between these nodes and their counterparts in $C^\ast$ are isomorphic by using the old auxiliary data. Additionally, at the splitting nodes, children subtrees are compared using \cref{claim:node-iso}. There is a constant number of cases, depending on the position of $h$, $a$ and $b^\ast$ relative to each other, which are all handled very similarly.
		\end{itemize}
\end{proof}

\end{document}